\documentclass[pra,aps,amsmath,amssymb,twocolumn,10pt]{revtex4-2}

\usepackage{graphicx}
\usepackage[colorlinks,bookmarks=true,citecolor=blue,linkcolor=blue,urlcolor=blue]{hyperref}

\usepackage{hyperref}

\usepackage{bm}
\usepackage{bbold}

\usepackage[resetlabels]{multibib}
\newcites{supp}{Supplementary References}

\renewcommand{\Re}{\mathop{\mathrm{Re}}}
\renewcommand{\Im}{\mathop{\mathrm{Im}}}
\newcommand{\sign}{\mathop{\mathrm{sign}}}
\newcommand{\nudyn}{\nu_\mathrm{dyn}}
\newcommand{\ceff}{c_\mathrm{eff}}
\newcommand{\kc}{k_\mathrm{c}}
\newcommand{\ii}{\mathrm{i}}
\newcommand{\ee}{\mathrm{e}}
\newcommand{\dd}{\mathrm{d}}

\usepackage{amsthm}
\newtheorem{theorem}{Theorem}
\newtheorem{proposition}[theorem]{Proposition}
\newtheorem{corollary}[theorem]{Corollary}

\begin{document}

\title{Spectral-topology-induced criticality in non-Hermitian fermionic metals}

\author{Ayan Banerjee$^{1,a,\dagger}$}

\author{Julius T. Gohsrich$^{1,2,b,\dagger}$}

\author{Flore K. Kunst$^{1,2,c}$}

\affiliation{$^{1}$Max Planck Institute for the Science of Light, 91058 Erlangen, Germany \\$^{2}$Department of Physics, Friedrich-Alexander-Universität Erlangen-Nürnberg, 91058 Erlangen, Germany \\ 
\emph{$^a$\,\href{mailto:ayan.banerjee@mpl.mpg.de}{ayan.banerjee@mpl.mpg.de} \qquad $^b$\,\href{mailto:julius.gohsrich@mpl.mpg.de}{julius.gohsrich@mpl.mpg.de} \qquad $^c$\,\href{mailto:flore.kunst@mpl.mpg.de}{flore.kunst@mpl.mpg.de}} \\
\emph{\phantom{\qquad \qquad \qquad}$^\dagger$These authors contributed equally. \phantom{\qquad \qquad \qquad}}}

\date{\today}

\begin{abstract}
Quantum matter emerges from the interplay of fluctuations, topology, and entanglement, which---in equilibrium---governs quantized transport, universal criticality, and topological classification.
Non-Hermitian systems, widely explored in platforms ranging from electric circuits to photonics, are intrinsically out-of-equilibrium, and display fundamentally new phenomena, including complex spectra, spectral winding, exceptional topology, and non-unitary dynamics.
A central challenge is understanding how the complex single-particle spectrum governs universal many-body behavior.
We introduce a symmetry-protected dynamical topological index derived directly from the complex spectrum.
Through the lens of algebraic topology, more specifically Morse theory, we identify critical points in the spectrum with topological defects, whose curvature and stability are protected under continuous deformations.
This links spectral geometry to many-body observables, unifying non-Hermitian band topology, entanglement, and transport.
We demonstrate that non-Hermitian quantum criticality in non-interacting systems is controlled by gain-and-loss-selected non-equilibrium steady states, which dynamically generate an emergent imaginary Fermi surface whose Fermi points host scale-invariant gapless modes with logarithmic entanglement scaling and algebraic correlations.
Our work establishes a unified framework for non-Hermitian quantum matter, connecting spectral topology to Morse theory, revealing a topological foundation of non-equilibrium quantum criticality.
\end{abstract}

\maketitle

\section{Introduction}

Mathematics has shown an \emph{unreasonable effectiveness} in unveiling deep physical structures and unifying different physical phenomena~\cite{wigner_unreasonable_1960}.
In condensed matter physics, this unifying power is most vividly expressed through topology, which has reshaped our understanding of quantized transport, robustness to disorder, and universal aspects of quantum criticality~\cite{haldane_nobel_2017,kosterlitz_nobel_2017}.
This synergy has been deepened by tools from knot theory~\cite{witten_quantum_1989}, graph theory~\cite{essam_graph_1971}, and algebraic topology~\cite{banerjee_tropical_2023}, revealing hidden patterns in quantum systems.
Over the past decades, topological band theory and quantum entanglement have emerged as complementary approaches to expose underlying geometric structures of collective quantum behavior~\cite{wen_colloquium_2017,sieberer_universality_2025}.
One central assumption underlying much of this progress is Hermiticity, which ensures real spectra and unitary dynamics in closed quantum systems.
Yet many experimentally relevant quantum platforms---from photonic structures~\cite{weidemann_topological_2020} and cold-atom setups~\cite{schreiber_observation_2015} to driven and open many-body systems~\cite{daley_quantum_2014}---are intrinsically non-Hermitian due to coupling with the environment, dissipation, or gain processes.
Such non-Hermitian systems~\cite{bender_real_1998,ashida_non-hermitian_2020,bergholtz_exceptional_2021,banerjee_non-hermitian_2023} admit complex spectra and show fundamentally new phenomenology, such as non-unitary dynamics~\cite{lee_topological_2019}, exceptional points~\cite{kato_perturbation_1966}, the non-Hermitian skin effect~\cite{yao_edge_2018,gohsrich_non-hermitian_2025-1}, and the breakdown of the conventional bulk--boundary correspondence~\cite{kunst_biorthogonal_2018,yao_edge_2018,zhang_correspondence_2020}.

These developments raise a fundamental question: How should quantum criticality be defined when spectra are complex and dynamics are non-unitary?
In non-Hermitian systems, criticality cannot be governed by equilibrium ground states, but instead must relate to non-equilibrium steady states selected dynamically by gain and loss~\cite{lee_topological_2019,panda_entanglement_2020,banerjee_chiral_2022}.
Remarkably, even in absence of unitary conformal invariance, such steady states can exhibit scale invariance, algebraic correlations, and logarithmic entanglement scaling even in one-dimensional systems~\cite{chang_entanglement_2020,panda_entanglement_2020,banerjee_chiral_2022}.
In particular, we put forward the idea that criticality in non-Hermitian systems arises when the steady-state dynamics host scale-invariant gapless modes.

In this work, we identify a topological principle that captures criticality.
We introduce a \emph{dynamical topological index} defined directly from the complex single-particle spectrum and protected by symmetry.
Rooted in algebraic topology, more specifically Morse theory, this index captures an intrinsic imbalance in the spectral curvature and provides a robust, topological count of emergent steady-state Fermi points.
The complex single-particle spectrum defines a mapping from the Brillouin zone to the complex energy plane, and if this mapping is smooth, endowing this closed spectral curve with a natural Morse structure.
Within this framework, critical points of the spectrum act as topological defects whose curvature and stability are protected under continuous deformations, as formalized by Morse theory.
As a result, the discussed dynamical topological index, and in extension algebraic topology, provides a natural language for connecting spectral geometry to emergent scale invariance, entanglement, and transport in non-Hermitian fermionic systems, revealing a topological origin of non-equilibrium criticality.

\clearpage
\section{Spectral topology in non-Hermitian systems}

One of the early questions in non-Hermitian physics was how to extend the topological classification, the famous ten-fold way~\cite{kitaev_periodic_2009,ryu_topological_2010}, due to spectra being complex.
A topological reclassification in line and point gaps provides such a generalization~\cite{shen_topological_2018,kawabata_symmetry_2019}.
In this context, a \emph{line gap} excludes an entire line in the complex energy plane, generalizing the notion of gaps in Hermitian systems, whereas a \emph{point gap} avoids a complex reference energy.
These distinct types of gaps give rise to different topological invariants and symmetry-protected edge phenomena.
Beyond line and point gap topology, other notions of spectral topology have been discussed in the literature~\cite{wang_topological_2021-1,konig_braid-protected_2023,yoshida_winding_2025,stalhammar_abelian_2025,montag_spectral_2026,lein_choosing_2020}.

While all these notions of topology are usually interpreted within a single-particle picture, their physical relevance and phenomenology is ultimately determined by non-unitary dynamics, and might be altered by incorporating many-particle effects.
In this work, we introduce a distinct spectral topological index, which relates the complex single-particle spectrum to physical quantities defined through non-equilibrium steady states, and thus, quantum criticality in fermionic many-body systems.

\subsection{Dynamical topological index}

Let us introduce the anticipated dynamical topological index, and discuss the intuition behind it after the proposition.

\begin{proposition}\label{prop:nudyn}
Let us consider the non-Hermitian Bloch Hamiltonian $\mathcal{H}(k)$ in one spatial dimension, with $k \in (-\pi,\pi]$, being time-reversal symmetric.
Assuming a single band model, $\mathcal{H}(k)$ is equal to the energy $E(k)$, and we set $f(k) = \Im E(k)$ and impose that $f$ is a Morse function and only has simple zeros. Then the dynamical topological index
\begin{equation}
    \label{eq:nudyn}
    \nudyn = - \sum_{\substack{\kc \ \mathrm{s.t. } \ f'(\kc)=0 \\ f(\kc)>0}}\sign\left[ \left. f''(k) \right|_{k = \kc} \right]
\end{equation}
is a non-negative integer measuring the net imbalance of extrema of $f(k)$ above zero.
\end{proposition}
Let us discuss the intuition behind our choices and assumptions, while we defer the proof to the Methods.
We consider a single-band model as simplest example showing all relevant physical features---the multi-band case reduces to this case with some technicalities as discussed below.
The requirement that the Bloch-Hamiltonian is time-reversal symmetric, i.e., $\mathcal{H}(-k)=\mathcal{T} \mathcal{H}^*(k) \mathcal{T}^\dagger$ with a unitary matrix $\mathcal{T}$ satisfying $\mathcal{T}\mathcal{T}^*=\pm\mathbb{1}$ and the identity matrix~$\mathbb{1}$, leads to the spectral constraint $E(-k)=E^*(k)$~\cite{kawabata_symmetry_2019}.
Therefore, the spectrum of $\mathcal{H}(k)$, $\{E\} =\{ E(k) \mid k\in(-\pi,\pi]\}$, is symmetric around the real axis, $\{E\} =\{E^*\}$, as seen in Fig.~\ref{fig:fig1}(a).
Physically, this spectral symmetry suggests a natural half-filling when introducing fermionic many-particle effects.
Technically, this enforces that $f(k)$ is odd, i.e., $f(-k)=-f(k)$ as illustrated in Fig.~\ref{fig:fig1}(b), which is used in the proof to connect maxima (minima) of $f$ above zero with minima (maxima) of $f$ below zero.
Coming to the multi-band case below, we will relax the constraint of time-reversal symmetry to work in a more general setting.
Requiring~$f$ to be a Morse function~\cite{morse_calculus_1934,milnor_morse_1963} is of technical nature:
It forces that all critical points $\kc$ of~$f$ are non-degenerate, i.e., all $\kc$ with $f'(\kc)=0$ satisfy $f''(\kc) \neq 0$, making all critical points correspond to either minima or maxima, and excludes the possibility that $f$ has saddle points.
This requirement allows us to connect our problem to the field of algebraic topology, more specifically, to Morse theory~\cite{morse_calculus_1934,milnor_morse_1963}.
The requirement to consider simple zeros only is again of technical nature, as we will connect to zero-crossings, and not zero-touchings, in the many-particle case.
Lastly, on the choice considering the imaginary parts of the energies:
In the single-particle picture, the systems long-time dynamics are dominated by the eigenstate corresponding to the energy with the largest imaginary part (or eigenstates when having multiple energies with the largest imaginary parts).
These are either the most amplified, or the slowest decaying eigenstates.
This intuition, that the imaginary part of the eigenstate dictates the long-time dynamics, will transfer to the many-particle case.

\begin{figure}[t]
    \centering
    \includegraphics[]{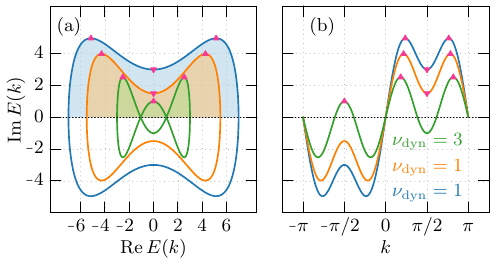}
    \caption{
    Spectral topology and dynamical topological index.
    (a) Different complex spectra (blue, orange, and green curves) for different parameters ($t_1=6$, $t_1=4.5$ and $t_1=2$, respectively) for the Hamiltonian in Eq.~\eqref{eq:H} with $l = r = 3$, $t_{-1} = 1$, $t_{\pm2} = 0$ and $t_{\pm3} = \pm 1$.
    The dynamical topological index is the difference between the number of maxima and number of minima in the upper half plane, and transitions occur when these numbers change.
    The extrema contributing to the dynamical topological index are highlighted by pink triangles, where the up-pointing and down-pointing triangles correspond to maxima and minima, respectively.
    The shaded regions contribute to the non-equilibrium steady state, which connects the dynamical topological index, stemming from the single-particle picture, to system quantities when considering many-particle effects.
    (b) Imaginary part of the dispersion $f(k) = \Im E(k)$ corresponding to (a) and the associated dynamical topological indices.
    Due to time-reversal symmetry, $f$ is an odd function.
    Again the extrema contributing to the dynamical topological index are highlighted by pink triangles.
    }
    \label{fig:fig1}
\end{figure}

Let us briefly summarize the crux of the proof from the Methods.
The argument combines three key ingredients:
(i)~$f$ being odd, (ii)~$f$ being a Morse function, and (iii)~$f$ being periodic.
Together, these ingredients constrain the Euler characteristic of $f$ and enforce that extrema of $f(k)$ necessarily occur in pairs, with each maximum (minimum) in the upper-half plane having a mirrored partner minimum (maximum) in the lower half plane, as shown in Fig.~\ref{fig:fig1}(a).
By isolating the extrema in the upper half plane, which relates to the populated states, the dynamical topological index captures a net spectral curvature imbalance.

\subsection{Example}

We consider a general non-interacting, one-dimensional, single-band tight-binding model with arbitrary long-range hopping with Hamiltonian
\begin{equation}
   H = \sum_{n=1}^N \sum_{j = -r}^{l} t_j \, c^\dagger_{n+j} c_n = \sum_k \mathcal{H}(k) \, c_k^\dagger c_k,
   \label{eq:H}
\end{equation}
where the $t_j$, with $j=-r,\ldots,l$, are hopping strengths, $r$ and $l$ the hopping ranges to the right and left, respectively, $c^\dagger_n$ ($c_n$) creates (annihilates) an excitation (later fermions) at site $n$, and $\mathcal{H}(k) = \sum_{j = -r}^{l} t_j \, \ee^{\ii k j}$ is the Bloch Hamiltonian.
This system is non-Hermitian if any $ t_j \neq t^*_{-j}$.
We assume real hopping strengths $t_j$, which makes the system time-reversal symmetric as required by Proposition~\ref{prop:nudyn}.

For now, we consider $l=r=3$ and parameters $t_{-1}=1$, $t_{\pm2}=0$, and $t_{\pm3}=\pm1$.
Figure~\ref{fig:fig1}(a) shows spectra of the system for different hopping strengths $t_1$.
Starting with the blue closed curve, we find two maxima and a single minimum in the upper half-plane, indicated by pink triangles, resulting in $\nudyn=1$.
One can clearly see the spectral symmetry around the real axis, as well as the fact that every maximum (minimum) in the upper half-plane has a corresponding minimum (maximum) in the lower half-plane.
Decreasing~$t_1$ slightly (orange curve) changes the details of the spectrum slightly, but retains the maxima and minima, and thus, the dynamical topological index.
Further decreasing $t_1$ (green curve) moves the minimum below zero.
Thus, the dynamical topological index is over all increased by two, leading to $\nudyn=3$.

This can also be understood by looking at the imaginary part of the dispersion $\Im E(k)$ in Fig.~\ref{fig:fig1}(b).
For the blue and orange curves one can find a single minimum above zero at $k=\pi/2$, and two maxima symmetric around this minimum.
For the green curve, the minimum at $k=\pi/2$ is below zero, thus not contributing to the dynamical topological index, but the maximum at $k=-\pi/2$---connected via time-reversal symmetry---is above zero, contributing to the index.
Importantly, the topological phase transition is marked by crossing two non-simple zeros at $k=\pm\pi/2$, as further illustrated in Fig.~\ref{fig:fig2}(a).

Finally, we also see in this example that the assumption of $f$~being Morse and only having simple zeros (per Proposition~\ref{prop:nudyn}) does in general not restrict the applicability of our framework to specific systems.
Instead, breaking any of these assumptions for a given system provides a mathematical mechanism to change the dynamical topological index and thus hints at a topological phase transition.

\begin{figure}[t]
\centering
    \includegraphics[trim={0 0 0 0.75cm},clip]{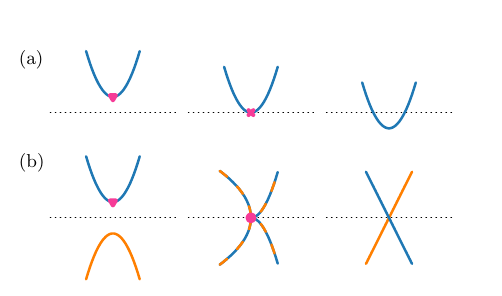}
    \caption{
    Mechanisms changing the dynamical topological index.
    The horizontal axis corresponds to $k$, the vertical axis to $\Im E(k)$, and the dashed horizontal line to the zero-axis $\Im E(k)=0$ as in Fig.~\ref{fig:fig1}(b).
    The minima contributing to the dynamical topological index are highlighted by pink down-pointing triangles.
    (a)~Under a continuous parameter change, the minimum above zero (left) can touch the zero-axis (center), which corresponds to a non-simple zero marked by a pink cross, and can move below zero (right), increasing the dynamical topological index by one.
    In general, by tuning through a non-simple zero, the dynamical topological index can change.
    (b) Additionally in multi-band systems, a minimum above zero and a maximum below zero (left) can touch at the zero axis and form a non-Hermitian degeneracy (center) marked by a pink dot, resulting in a crossing (right) under a continuous parameter variation, over all decreasing the dynamical topological index by one.
    In general, multiple extrema can `annihilate' or can be created by such exceptional point transitions and change the dynamical topological index.}
    \label{fig:fig2}
\end{figure}

\subsection{Comments on multi-band systems}

The intuition of the single-band case carries over to the multi-band case, where the spectrum of a non-Hermitian system consists of one or more closed curves in complex plane.
We decompose the spectrum into its different contributions and treat every contribution as in the single-band case.
Thus, we generalize the dynamical topological index as follows:

\begin{corollary}
Let us again consider a non-Hermitian Bloch Hamiltonian $\mathcal{H}(k)$ being time-reversal symmetric, but this time we consider a multi-bands system.
We require $\mathcal{H}(k)$ to be diagonalizable for all $k$, and diagonalizing it results in the bands $E_j(k)$, where $j$ labels the bands.
We define $f_j(k)= \Im E_j(k)$ and require all $f_j$ to be Morse functions and only have simple zeros.
Then, the dynamical topological index is
\begin{equation}
    \label{eq:nudyn2}
    \nudyn = - \sum_j\sum_{k_{\mathrm{c},j}}\sign\left[ \left. f_j''(k) \right|_{k = k_{\mathrm{c},j}} \right],
\end{equation}
where the sum over $k_{\mathrm{c},j}$ runs over all extrema of the respective band $j$.
\end{corollary}
The extra assumption in the multi-band case is that $\mathcal{H}(k)$ must be diagonalizable for all $k$, which is by definition always possible away from exceptional points~\cite{kato_perturbation_1966,ashida_non-hermitian_2020,bergholtz_exceptional_2021}.
At an exceptional point, not all $f_j$ may be smooth, breaking the assumption that $f_j$ is Morse, cf. Supplementary Information (SI).

However, an exceptional point can be seen as a mechanism which can change the dynamical topological index.
For example, an exceptional point can `annihilate' a minimum of a band above zero with a maximum below zero, changing the index.
The mechanisms changing the dynamical topological index, namely, an extremum crossing the zero axis through a non-simple zero, and such an exceptional point transition, are shown in Fig.~\ref{fig:fig2}.

We furthermore can relax the constraint of time-reversal symmetry to the spectral constraint $\{E\}=\{E^*\}$, where $\{E\} = \{E_j(k)\mid k\in(-\pi,\pi],\forall j\}$.
This essentially enforces that the entire spectrum is symmetric around the real axis.
This spectral constraint, for example, follows from time-reversal symmetry as before, or pseudo-Hermiticity, which is a similarity of which $\mathcal{PT}$-symmetry and pseudo-Hermitian symmetry are special cases~\cite{montag_essential_2024-1}.

In the SI, we provide single- and two-band models, which satisfy this spectral constraint, and also feature the change of dynamical topological index by crossing an exceptional point.
For the remaining examples in the manuscript, we will consider time-reversal symmetric single-band models, i.e., given by Eq.~\eqref{eq:H} with real hoppings.

\section{Many-particle effects and non-unitary dynamics}

In Hermitian systems including many-particle effects, in particular, non-interacting many-body systems, the notion of equilibrium stems from real spectra which imply stable quasiparticles and time-independent statistical ensembles.
Non-Hermitian systems, however, are in general out-of-equilibrium as quasiparticles are not stable due to finite lifetimes introduced by the imaginary parts of many-body energies.
In such a non-equilibrium description, the dynamical topological index will manifest itself dynamically through non-unitary time evolution.
Out of equilibrium, the relevant long-time many-body state is the non-equilibrium steady state, which selectively amplifies single-particle modes with positive imaginary energies, defining an imaginary Fermi surface absent in equilibrium physics.
In particular, the non-equilibrium steady state is controlled by the imaginary Fermi surface, where long-lived gapless excitations emerge, which define a non-Hermitian fermionic metal.
Consequently, the dynamical topological index, defined through the single-particle spectrum, directly determines universal many-body properties such as entanglement scaling, steady-state currents and correlations.

\subsection{Non-equilibrium steady state}

Let us now construct the non-equilibrium steady state of a non-Hermitian free-fermion system for finite filling as in Refs.~\cite{panda_entanglement_2020,banerjee_chiral_2022}.
To this end, we consider a generic non-Hermitian quadratic fermionic Hamiltonian in real space of the form $H = \phi^\dagger \mathcal{H}\, \phi$,
where $\phi^\dagger = (c_1^\dagger, \ldots, c_M^\dagger)$ and $\phi = (c_1, \ldots, c_M)^T$ are vectors of fermionic creation and annihilation operators, respectively.
These operators satisfy the canonical anti-commutation relations $\{c_m, c_{m'}^\dagger\} = \delta_{mm'}$.
As $\mathcal{H}$ is an $M\times M$ non-Hermitian single-particle Hamiltonian, its diagonalization requires a biorthogonal basis.
Its right and left eigenvectors are
\begin{align}
    \mathcal{H}\, |R_\alpha\rangle &= E_\alpha |R_\alpha\rangle,
    &
    \langle L_\alpha|\, \mathcal{H} &= E_\alpha \langle L_\alpha|,
\end{align}
respectively, where $E_\alpha$ are, in general, complex energies.
The left and right eigenvectors are biorthonormalized as $\langle L_\alpha | R_\beta \rangle = \delta_{\alpha\beta}$. Using this biorthogonal basis, we introduce the transformed fermionic operators $ \tilde{c}_\alpha^\dagger = \sum_m \langle i | R_\alpha \rangle \, c_m^\dagger$ and $\tilde{c}_\alpha = \sum_m \langle L_\alpha | m \rangle \, c_m$, which preserve the fermionic anti-commutation relations as a consequence of the biorthonormality of the eigenvectors.
In terms of these operators, the Hamiltonian becomes diagonal as $ H = \sum_{\alpha=1}^{M} E_\alpha \, \tilde{c}_\alpha^\dagger \tilde{c}_\alpha$.

With that, any many-body (right) eigenstate at finite filling with $P$ fermions is constructed as $ |\Psi_p\rangle = \tilde{c}^\dagger_{p_1}\tilde{c}^\dagger_{p_2} \cdots \tilde{c}^\dagger_{p_P} |\Omega\rangle$, where $|\Omega\rangle$ is the vacuum state, and $p=(p_1,\ldots,p_P)$ is the multi-particle index enumerating the configuration of $p_j$ over $N$ sites.
This eigenstate has many-body energy $\mathcal{E}_p=\sum_{j=1}^{P} E_{p_j}$, and thus evolves under the non-unitary time evolution as $|\Psi_p(t)\rangle = \ee^{- \ii \mathcal{E}_p t} |\Psi_p(0)\rangle$.
Importantly, non-unitary evolution necessitates renormalization of the many-body state during time evolution.

Consequently, under such a non-unitary time evolution, the system at finite filling asymptotically relaxes into a non-equilibrium steady state with a total energy $\mathcal{E}_{\tilde{p}}$ with maximum total imaginary energy out of all many-body states formed by occupying $P$ of the $N$ single-particle modes in configuration~$\tilde{p}$.
For later reference, we define this non-equilibrium steady state (labeled by NESS)~\cite{panda_entanglement_2020,banerjee_chiral_2022} at finite filling as
\begin{equation}
    | \Psi_\mathrm{NESS} \rangle \sim  |\Psi_{\tilde{p}}(t\rightarrow \infty)\rangle,
\end{equation} 
which is uniquely selected by the complex spectrum of~$\mathcal{H}$.

\subsection{Non-Hermitian Fermi physics}

This construction of the non-equilibrium steady state justifies its interpretation as a dynamical Fermi sea formed by the single-particle modes with the largest imaginary energies.
In particular, the imaginary Fermi surface $\Im E (k) = 0$ naturally emerges as the boundary between dynamically suppressed and dynamically amplified modes.
Time-reversal symmetry, discussed above, forces the spectrum to be symmetric around the real axis, i.e., $\Im E (k) = 0 $, which suggests half-filling as a natural choice for a finite-filled system.
Then, Fermi points are the momenta $k_\mathrm{F}$ at which the highest-occupied single-particle modes of the non-equilibrium steady state touch this boundary, i.e., $\Im E(k_\mathrm{F})=0$, thereby separating dynamically amplified (occupied) from dynamically suppressed (unoccupied) modes.

Along the imaginary Fermi surface, excitations above the non-equilibrium steady states correspond to modes with finite group velocity and form robust unidirectional chiral channels. 
The Nielsen--Ninomiya theorem enforces fermion doubling~\cite{nielsen_no-go_1981}:
Chiral modes must appear in pairs of opposite chirality, forbidding net chirality.
While the Nielsen--Ninomiya theorem formally holds in non-Hermitian systems, our systems dynamically evade it, as the distinction into amplifying and decaying modes provides a mechanism for asymmetry.
Let us connect the dynamical topological index to this Fermi physics, before coming back to this point.

\begin{corollary}
The number of Fermi points $n_\mathrm{F}$ is twice the dynamical topological index, i.e.,
\begin{equation}
    \label{eq:nF2nudyn}
    n_\mathrm{F} =2 \nudyn.
\end{equation}
\end{corollary}
As shown in the Methods, this is ensured by the intermediate value theorem and the requirement that $f$ only has simple zeros.
This relation originates in the interplay of time-reversal symmetry, non-unitary dynamics, and the properties of the non-equilibrium steady state:
Each unit of dynamical topological index forces the complex dispersion to cross $\Im E(k)=0$ twice, thereby producing a pair of symmetry-related Fermi points at $\pm k_\mathrm{F}$, cf. Fig.~\ref{fig:fig1}.
Although this pairing reflects a remnant of fermion doubling, non-unitary dynamics selects amplifying modes, allowing these crossings to act as gapless and dynamically protected excitations above the non-equilibrium steady state.
In this sense, $\nudyn$ quantifies a dynamical violation of the Nielsen--Ninomiya theorem, stabilized by non-unitary evolution rather than band topology alone.
Thus,~$\nudyn$ provides a robust, topologically protected measure of chiral asymmetry.
Over all, this allows interpreting the system as a \emph{symmetry-protected fermionic chiral metal}~\cite{ying_symmetry-protected_2018,bessho_nielsen-ninomiya_2021,banerjee_chiral_2022}. 

Furthermore, one can view the dynamical topological index as a quantifier of the net curvature-weighted imbalance of these extrema.
This asymmetry gives rise to an intrinsic topological tension between spectral curvature and band continuity, which is resolved through the nucleation of paired imaginary Fermi points.

\section{Entanglement scaling and effective central charge}

These dynamically protected Fermi points, characterized by real group velocities, support long-lived propagating modes and contribute to robust steady-state currents.
The same Fermi points also govern logarithmic entanglement scaling, which is characterized by an \textit{emergent} effective central charge.
To characterize quantum correlations in the non-equilibrium steady state, we compute the fermionic correlation matrix restricted to a contiguous subsystem $A$ of size $\ell$.
The correlation matrix has elements $C_{i,j} = \langle \Psi_\mathrm{NESS} | c_i^\dagger c_j | \Psi_\mathrm{NESS} \rangle$.
Its eigenvalues $e_i$ fully determine the entanglement entropy via Peschel's prescription~\cite{peschel_calculation_2003,calabrese_entanglement_2009}, as $S_\ell = -\sum_i \left[ e_i \log e_i + (1-e_i)\log(1-e_i) \right]$ quantifies the entanglement between the retained and traced-out regions.
The scaling behavior of $S_\ell$ serves as a diagnostic of criticality of the non-equilibrium steady state.
Under periodic boundary conditions, and in analogy with equilibrium conformal field theories, the entanglement entropy exhibits logarithmic scaling
\begin{equation}
    S_\ell = \frac{\ceff}{3}
    \log\!\left[\frac{N}{\pi}\sin\!\left(\frac{\pi\ell}{N}\right)\right]
    + \mathrm{const.},
\end{equation}
where $\ceff$ is the so-called effective central charge~\cite{peschel_calculation_2003,calabrese_entanglement_2009,chang_entanglement_2020}.
This logarithmic scaling is illustrated in the inset of Fig.~\ref{fig:fig4}.
In non-Hermitian systems, $\ceff$ can differ from the unitary conformal field theory value and encodes universal features of the entanglement structure of the non-equilibrium steady state~\cite{chang_entanglement_2020}.
Physically, the logarithmic scaling originates from gapless modes stemming from the Fermi points.
Although the non-equilibrium steady state is dominated by modes with positive imaginary energies, it is these modes that control universal correlations, in direct analogy with the role of Fermi points in Hermitian critical systems.
As shown below, their number is fixed by the dynamical topological index introduced earlier:
\begin{corollary}
The effective central charge governing the entanglement scaling of the non-equilibrium steady state at half filling equals the dynamical topological index, i.e.,
\begin{equation}
    \ceff = \nudyn.
\end{equation}
\end{corollary}
Each unit increase of the dynamical topological index $\nudyn$ enforces the appearance of two imaginary Fermi points (cf. Eq.~\eqref{eq:nF2nudyn}).
Within low-energy field theory, pairs of Fermi points form Dirac-like gapless modes whose contributions add to the central charge governing logarithmic entanglement scaling in one-dimensional critical systems ~\cite{calabrese_entanglement_2004,calabrese_entanglement_2009}.
Consequently, the total central charge is fixed by the number of dynamically protected gapless modes.
Thus, $\nudyn$ serves as a topologically protected measure of entanglement scaling in the system.

\section{Transport and correlations}

Beyond entanglement scaling, transport and correlations provide a dynamical characterization of the non-equilibrium steady state as a non-Hermitian fermionic metal.

\vspace{-0.5cm}

\subsection{Persistent current}

The metallic behavior is tied to the above discussed Fermi physics, which stabilizes gapless, current-carrying modes absent in bosonic systems.
As discussed above, our system dynamically evades the Nielsen--Ninomiya theorem, thereby dynamically breaking time-reversal symmetry, by selectively amplifying one species of chiral modes.
These current-carrying modes give rise to a finite persistent current in the nonequilibrium steady state.

Assuming a single-band model and assigning unit charge to each fermion, the persistent current~\cite{zhang_correspondence_2020} is
\begin{equation}
    J 
    = \int_{-\pi}^{\pi} \! \! \! n(E) \partial_k E \, \dd k 
    = \oint_\mathrm{BZ} \! \! \! n(E)\, \dd E,
    \label{eq:current}
\end{equation}
where $E=E(k)$ traces the complex spectrum as $k$ spans the Brillouin zone (BZ).
The occupation function $n(E)$ in the non-Hermitian case,
$n(E) = \{1+ \exp{(-\beta[\Im E(k)-\mu)]}\}^{-1}$
%
%
with inverse temperature $\beta = 1/T$ and chemical potential \mbox{$\mu = 0$} at half-filling~\cite{zhang_correspondence_2020,banerjee_chiral_2022}, depends only on the imaginary part of the energy and preferentially populates amplifying modes.
For simplicity, we set $T \to 0$ so that $n(E)$ is a step function that only selects states with positive imaginary part.
Associated to the persistent current there are also the number densities $n_{\mathrm{L},\mathrm{R}}$, i.e., the fraction of left- and right-moving modes, respectively, so that $n_L+n_R=1$.
They are given by $n_{\mathrm{L},\mathrm{R}}=|\mathcal{K}_{\mathrm{L},\mathrm{R}}|/(|\mathcal{K}_{\mathrm{L}}|+|\mathcal{K}_{\mathrm{R}}|)$, where $\mathcal{K}_{\mathrm{L},\mathrm{R}}=\{ k \mid \Re\left[\partial_k E(k)\right]\lessgtr 0 \wedge \Im E(k)\geq0\}$.
While the current $J$ is continuous under continuous parameter changes, we find it is not differentiable at points where the spectral topology changes, i.e., we find non-analytic behavior at topological phase transitions.
Physically, this reflects the sudden appearance or disappearance of Fermi points, whose associated current-carrying modes contribute to the persistent current.
Therefore, the persistent current provides a sensitive and experimentally accessible probe of transitions of the spectral topology.

\begin{figure}[t]
    \centering
    \includegraphics[]{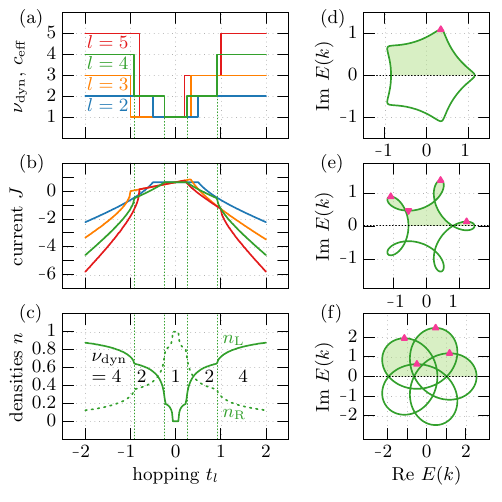}
    \caption{
    Spectral topological phase transitions and transport phenomena in a generalized Hatano--Nelson model.
    (a) Dynamical topological index $\nudyn$ and central charge $\ceff$ as function of $t_l$ for different $l$, showing topological phase transitions at distinct transition points.
    The green-dashed vertical lines mark the hopping strengths~$t_l$ at which the dynamical topological index for the model with $l=4$ changes.
    These lines are also present in (b,d) to highlight qualitative changes at these hopping strengths due to these phase transitions for the model with $l=4$.
    (b) Total current as a function of $t_l$ with same $l$ as in (a), showing the non-analytic behavior at the same  parameters where the topological phase transitions happen in~(a).
    (c) Number densities of left- ($n_\mathrm{L}$) and right-movers ($n_\mathrm{R}$) as functions of $t_l$ for $l=4$, showing the same non-analytic behavior.
    \mbox{(d-f)}~Complex spectra for $l=4$ with $t_4=0.15$, $t_4=0.45$ and $t_4=1.55$, respectively, illustrating the evolution of spectral topology across different phases.
    In all panels, we have set $t_{-1}=1$ and $t_j=0$ for all $j\neq l,-1$.
    }
    \label{fig:fig3}
\end{figure}

This behavior is exemplified in Fig.~\ref{fig:fig3}.
As a model, we again have used the one-dimensional single-band Hamiltonian, Eq.~\eqref{eq:H}, where we vary the hopping range to the left $l$ and the associated hopping strength $t_l$, while fixing $r=1$, $t_{-1}=1$ and set all the other hopping strengths to zero, corresponding to the Bloch Hamiltonian $\mathcal{H}(k) = t_l \ee^{\ii kl}+\ee^{-\ii k}$.
This model has recently been investigated as a generalized version of the Hatano--Nelson model~\cite{hatano_localization_1996,gohsrich_exceptional_2024}, and it exhibits a generalized chiral symmetry~\cite{marques_generalized_2022}, resulting in rotational symmetric complex spectra.

Figure~\ref{fig:fig3}(a) shows the dynamical topological index and the identical central charge as a function of~$t_l$ for different~$l$.
Jumps are clearly visible at different~$t_l$, corresponding to changes in the dynamical topological index and effective central charge.
In Fig.~\ref{fig:fig3}(b), one can see the persistent current~$J$ for the corresponding~$t_l$ and~$l$.
At the jumps in (a),~$J$ shows non-analytic behavior in (b).
For $l=4$, the spectral transitions occur at $t_4=\pm1/4$ and $t_4=\pm\sqrt{27/32} \approx \pm 0.92$.
Additionally, we show the number densities of the left- and right movers, $n_{\mathrm{L},\mathrm{R}}$ respectively, for $l=4$ in Fig.~\ref{fig:fig3}(c).
Beyond the non-analytic points stemming from transitions in the spectral topology, there are additional non-analytic points close to zero where the number densities become constant.
\mbox{Figures \ref{fig:fig3}(d-f)} show different complex spectra for $t_4=0.15$, $t_4=0.45$ and $t_4=1.55$, respectively.
Beyond the apparent rotational symmetry of the spectra, stemming from the generalized chiral symmetry, one can clearly observe the changes of the spectral topology.

\begin{figure}[t]
\centering
    \includegraphics[]{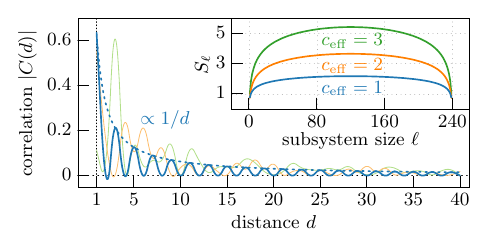}
    \caption{
    Correlations and entanglement entropy.
    The system exhibits power-law correlations $\propto 1/d$ with distinct interference patterns as a function of distance~$d$, indicating a gapless critical phase and scale invariance.
    We quadratically interpolate between the different integer $d$ for better visibility.
    Inset:
    Corresponding entanglement entropies $S_\ell$ showing logarithmic scaling as function of subsystem size $\ell$ and distinct central charges.
    We have set $t_{\pm1}=1\pm t$, $t_{\pm2}=\pm0.2$, $t_{\pm3}=\mp1.9$, $t_{\pm4}=\pm1$, with $t=11$ (blue), $t=6$ (orange), $t=1$ (green).
    }
    \label{fig:fig4}
\end{figure}

\subsection{Correlations}

In Hermitian many-body systems, long-range correlations provide a direct probe of criticality and scale invariance.
In gapped phases, correlations decay exponentially, reflecting the presence of a finite correlation length.
In contrast in gapless phases, algebraic (power-law) decay signals the absence of a characteristic length scale and is a defining feature of criticality, including metallic and conformal phases.
In one-dimensional fermionic systems, such behavior is typically governed by the structure of the Fermi surface:
Gapless excitations generate long-wavelength coherence, while the number and the arrangement of Fermi points control oscillations and decay exponents of correlation functions.
In our non-Hermitian framework, this paradigm is reshaped by non-unitary dynamics, the relevance of the imaginary Fermi surface, and the emergence of the non-equilibrium steady state.
Here, scale invariance arises not from the gap closing, but from the spectral topology of the complex dispersion with the emergent Fermi points and the dynamical selection of modes with maximal imaginary energy.
Consequently, correlation functions encode not only coherence but also the underlying dynamical topology that organizes the steady state.

To understand this mechanism in more detail, we consider the fermionic two-point correlator of the non-equilibrium steady state $C(d)=C_{i,i+d}$ for any $i$, where $C_{i,j}$ are the correlation matrix elements from above.
Topological transitions in the spectrum reshape these interference patterns, thereby controlling both correlations and entanglement through the dynamical topological index.
In regimes where the system exhibits multiple Fermi points, $|C(d)|$ displays oscillatory behavior with algebraic decay, reflecting interference between coherent contributions from different momentum modes as illustrated in Fig.~\ref{fig:fig4}.
When varying system parameters, the number and position of these Fermi points undergo topological transitions, which can significantly modify the interference pattern---enhancing or suppressing correlations at specific distances through constructive or destructive interference.
Thus, these transitions lead to notable changes in the correlation profile and entanglement entropy, governed by the dynamical topological index.

\section{Conclusions}

Our work opens several promising avenues for future exploration.
While we have formulated a dynamical-topological framework for one-dimensional non-Hermitian systems, a natural and highly nontrivial extension lies in generalizing this construction to higher-dimensional systems.
Such generalizations would allow discussing the inclusion of fluxes, disorder, interactions~\cite{zhang_symmetry_2022-1,kawabata_many-body_2022-1} and symmetry breaking~\cite{zhang_symmetry_2022-1}.
In this context, it would be particularly interesting to explore links to self-healing and non-unitary dynamical stabilization phenomena~\cite{longhi_self-healing_2022}.
Looking ahead, our approach may be extended to encompass other ramified non-Hermitian symmetry classes, connections to $K$-theoretic classifications, and a deeper understanding of dynamical topology.

As a promising experimental avenue, our predictions might be accessible in state-of-the-art cold-atom~\cite{schreiber_observation_2015} and photonic platforms~\cite{weidemann_topological_2020}, where non-Hermitian control~\cite{ashida_non-hermitian_2020,bergholtz_exceptional_2021}, tunable dissipation~\cite{daley_quantum_2014}, and real-time measurement of correlations~\cite{fitzpatrick_observation_2017} are now routinely achievable.
These platforms offer a route toward observing unconventional phases and critical behavior in non-Hermitian quantum matter.

More broadly, by bridging Mores-theoretic structures with physically realizable non-Hermitian dynamics, our work contributes to the growing dialogue between modern mathematics and non-Hermitian physics.
Given the rapid surge of interest in these fields, we hope that our results will stimulate further interdisciplinary developments.

\section*{Acknowledgments}
A.B. acknowledges a related collaboration with Adhip Agarwala, Suraj S. Hegde and Awadhesh Narayan, and thanks Sanchayan Banerjee and Anumita Bose for illuminating discussions.
All authors thank Anton Montag for fruitful discussions related to symmetries and similarities.
A.B., J.T.G., and F.K.K. acknowledge funding from the Max Planck Society Lise Meitner Excellence \mbox{Program 2.0}.
A.B., J.T.G. and F.K.K. also acknowledge support from the European Union's ERC Starting Grant ``NTopQuant'' (101116680).
The views expressed are those of the authors and do not necessarily reflect those of the European Union or the ERC.

\section*{Author Contributions}
\noindent
A.B. initiated the project.
A.B. and J.T.G. developed the formalism.
A.B. and J.T.G. wrote the manuscript with input from F.K.K..
F.K.K. supervised the project.

\section*{Data Availability}
\noindent
All study data are included in the main text and supplementary information.

\section*{Competing Interests}
\noindent
The authors declare no competing interests.

\section*{Methods}
\subsection*{Proof of Proposition 1}

\begin{proof}
The imaginary part of the dispersion, $f(k)=\Im E(k)$, over the Brillouin zone $k\in(-\pi,\pi]$ is a $2\pi$-periodic real-valued function (as we identify $f(\pi)=f(-\pi)$).
Thus, it is topologically a circle $S^1$, i.e., $f:S^1\to S^1$.
On the one hand, as such, it has non-vanishing Betti numbers $b_0=1$ and $b_1=1$, resulting in a vanishing Euler characteristic $\chi(S^1)=b_0-b_1=0$.
On the other hand, as we require $f$ to be Morse, Morse theory additionally gives $\chi(S^1)=C_0-C_1$, where $C_0$ and $C_1$ denote the numbers of minima and maxima of $f$, respectively.
Together, this yields $C_0=C_1$.
This equality is purely topological and remains stable under smooth symmetry-preserving deformations.

Time-reversal symmetry imposes $f(-k)=-f(k)$, which enforces pairing of extrema $k_\mathrm{c}\leftrightarrow -k_\mathrm{c}$ with opposite Morse index, i.e., $f''(-k_\mathrm{c})=-f''(k_\mathrm{c})$.
Thus, extrema above and below the real axis are paired with opposite Morse index.

This also holds in the more general case of having $\{E\}=\{E^*\}$ instead of time-reversal symmetry, which imposes that for each extremum $k_\mathrm{c}$, there exists another extremum $\bar{k}_\mathrm{c}$ with opposite Morse index.
However, in this more general case, \mbox{$\bar{k}_\mathrm{c} \neq -k_\mathrm{c}$}.
\end{proof}

\subsection*{Comments on the multi-band case, Corollary 2}

In the multi-band case, the spectrum in general consists of multiple closed curves, as exemplified in the SI.
Importantly, multiple energy bands can make up such a closed spectral curve.
For each such closed spectral curve, we go to an extended Brillouin zone scheme and combine the appropriate bands $E_{j_1}(k),\ldots,E_{j_\nu}(k)$ into a parametrization of the closed spectral curve $\tilde{E}_{\tilde\jmath}(\tilde{k})$, where $\nu$ is the number of combined bands and $\tilde{k} \in (-\pi\nu,\pi\nu]$.
This procedure allows us to define $\tilde{f}_{\tilde\jmath}(\tilde{k})= \Im \tilde{E}_{\tilde\jmath}(\tilde{k})$, which is a parametrization of the imaginary part of the spectral curve, essentially combining the imaginary parts of the bands $f_{j_1}(k),\ldots,f_{j_\nu}(k)$.

Each of the $\tilde{f}_{\tilde\jmath}$ define a map $\tilde{f}_{\tilde\jmath}:S^1\to S^1$.
For each such mapping, the proof of Proposition 1 holds as we require that the $f_j$, and by extension the $\tilde{f}_{\tilde\jmath}$, are Morse functions.
Hence, we can associate a dynamical topological index to each $\tilde{f}_{\tilde\jmath}(\tilde{k})$ via Eq.~\eqref{eq:nudyn}, and summing over all $\tilde\jmath$ yields
\begin{equation}
    \nudyn = - \sum_{\tilde\jmath}\sum_{\tilde{k}_{\mathrm{c},\tilde\jmath}}\sign\left[ \left. \tilde{f}_{\tilde\jmath}''(\tilde{k}) \right|_{\tilde{k} = \tilde{k}_{\mathrm{c},\tilde\jmath}} \right],
\end{equation}
where the sum over $\tilde{k}_{\mathrm{c},\tilde\jmath}$ runs over all extrema of $\tilde{f}_{\tilde\jmath}$.
Equation \eqref{eq:nudyn2} immediately follows by expressing the $\tilde{f}_{\tilde\jmath}$ in terms of $f_j$.

\subsection*{Proof of Corollary 3}

As discussed above, we can reduce the multi-band case to the single-band case.
Thus, we formulate the proof in the context of the single-band case.

\begin{proof}
The dynamical topological index counts the excess of maxima $C_1^+$ over minima $C_0^+$ in the upper half-plane of the spectrum (indicated by the superscript $+$), i.e., $\nudyn=C_1^+-C_0^+$.
As the total number of extrema is fixed by $\chi(S^1)=0$, any such imbalance must be compensated by extrema in the lower half-plane of the spectrum, which forces zeros of $f$.
As $f$ is smooth, periodic, and only has simple zeros, the intermediate value theorem predicts that for each unit of curvature in the upper half-plane there are two real-energy crossings related by symmetry.
Hence, Eq.~\eqref{eq:nF2nudyn} holds.
\end{proof}

Therefore, the existence of these dynamically generated Fermi points is topologically protected by $\chi(S^1)=0$, and the number of Fermi points can only be changed by breaking the symmetry, by passing through a non-Morse degeneracy, or by passing through a non-simple zero.

\def\bibsection{\section*{\refname}} 

\bibliographystyle{apsrev4-2}
\bibliography{references}

\begin{thebibliography}{50}%
\makeatletter
\providecommand \@ifxundefined [1]{%
 \@ifx{#1\undefined}
}%
\providecommand \@ifnum [1]{%
 \ifnum #1\expandafter \@firstoftwo
 \else \expandafter \@secondoftwo
 \fi
}%
\providecommand \@ifx [1]{%
 \ifx #1\expandafter \@firstoftwo
 \else \expandafter \@secondoftwo
 \fi
}%
\providecommand \natexlab [1]{#1}%
\providecommand \enquote  [1]{``#1''}%
\providecommand \bibnamefont  [1]{#1}%
\providecommand \bibfnamefont [1]{#1}%
\providecommand \citenamefont [1]{#1}%
\providecommand \href@noop [0]{\@secondoftwo}%
\providecommand \href [0]{\begingroup \@sanitize@url \@href}%
\providecommand \@href[1]{\@@startlink{#1}\@@href}%
\providecommand \@@href[1]{\endgroup#1\@@endlink}%
\providecommand \@sanitize@url [0]{\catcode `\\12\catcode `\$12\catcode `\&12\catcode `\#12\catcode `\^12\catcode `\_12\catcode `\%12\relax}%
\providecommand \@@startlink[1]{}%
\providecommand \@@endlink[0]{}%
\providecommand \url  [0]{\begingroup\@sanitize@url \@url }%
\providecommand \@url [1]{\endgroup\@href {#1}{\urlprefix }}%
\providecommand \urlprefix  [0]{URL }%
\providecommand \Eprint [0]{\href }%
\providecommand \doibase [0]{https://doi.org/}%
\providecommand \selectlanguage [0]{\@gobble}%
\providecommand \bibinfo  [0]{\@secondoftwo}%
\providecommand \bibfield  [0]{\@secondoftwo}%
\providecommand \translation [1]{[#1]}%
\providecommand \BibitemOpen [0]{}%
\providecommand \bibitemStop [0]{}%
\providecommand \bibitemNoStop [0]{.\EOS\space}%
\providecommand \EOS [0]{\spacefactor3000\relax}%
\providecommand \BibitemShut  [1]{\csname bibitem#1\endcsname}%
\let\auto@bib@innerbib\@empty
\bibitem [{\citenamefont {Wigner}(1960)}]{wigner_unreasonable_1960}%
  \BibitemOpen
  \bibfield  {author} {\bibinfo {author} {\bibfnamefont {E.~P.}\ \bibnamefont {Wigner}},\ }\href {https://doi.org/10.1002/cpa.3160130102} {\bibfield  {journal} {\bibinfo  {journal} {Communications on Pure and Applied Mathematics}\ }\textbf {\bibinfo {volume} {13}},\ \bibinfo {pages} {1} (\bibinfo {year} {1960})}\BibitemShut {NoStop}%
\bibitem [{\citenamefont {Haldane}(2017)}]{haldane_nobel_2017}%
  \BibitemOpen
  \bibfield  {author} {\bibinfo {author} {\bibfnamefont {F.~D.~M.}\ \bibnamefont {Haldane}},\ }\href {https://doi.org/10.1103/RevModPhys.89.040502} {\bibfield  {journal} {\bibinfo  {journal} {Reviews of Modern Physics}\ }\textbf {\bibinfo {volume} {89}},\ \bibinfo {pages} {040502} (\bibinfo {year} {2017})}\BibitemShut {NoStop}%
\bibitem [{\citenamefont {Kosterlitz}(2017)}]{kosterlitz_nobel_2017}%
  \BibitemOpen
  \bibfield  {author} {\bibinfo {author} {\bibfnamefont {J.~M.}\ \bibnamefont {Kosterlitz}},\ }\href {https://doi.org/10.1103/RevModPhys.89.040501} {\bibfield  {journal} {\bibinfo  {journal} {Reviews of Modern Physics}\ }\textbf {\bibinfo {volume} {89}},\ \bibinfo {pages} {040501} (\bibinfo {year} {2017})}\BibitemShut {NoStop}%
\bibitem [{\citenamefont {Witten}(1989)}]{witten_quantum_1989}%
  \BibitemOpen
  \bibfield  {author} {\bibinfo {author} {\bibfnamefont {E.}~\bibnamefont {Witten}},\ }\href {https://doi.org/10.1007/BF01217730} {\bibfield  {journal} {\bibinfo  {journal} {Communications in Mathematical Physics}\ }\textbf {\bibinfo {volume} {121}},\ \bibinfo {pages} {351} (\bibinfo {year} {1989})}\BibitemShut {NoStop}%
\bibitem [{\citenamefont {Essam}(1971)}]{essam_graph_1971}%
  \BibitemOpen
  \bibfield  {author} {\bibinfo {author} {\bibfnamefont {J.~W.}\ \bibnamefont {Essam}},\ }\href {https://doi.org/10.1016/0012-365X(71)90009-4} {\bibfield  {journal} {\bibinfo  {journal} {Discrete Mathematics}\ }\textbf {\bibinfo {volume} {1}},\ \bibinfo {pages} {83} (\bibinfo {year} {1971})}\BibitemShut {NoStop}%
\bibitem [{\citenamefont {Banerjee}\ \emph {et~al.}(2023{\natexlab{a}})\citenamefont {Banerjee}, \citenamefont {Jaiswal}, \citenamefont {Manjunath},\ and\ \citenamefont {Narayan}}]{banerjee_tropical_2023}%
  \BibitemOpen
  \bibfield  {author} {\bibinfo {author} {\bibfnamefont {A.}~\bibnamefont {Banerjee}}, \bibinfo {author} {\bibfnamefont {R.}~\bibnamefont {Jaiswal}}, \bibinfo {author} {\bibfnamefont {M.}~\bibnamefont {Manjunath}},\ and\ \bibinfo {author} {\bibfnamefont {A.}~\bibnamefont {Narayan}},\ }\href {https://doi.org/10.1073/pnas.2302572120} {\bibfield  {journal} {\bibinfo  {journal} {Proceedings of the National Academy of Sciences}\ }\textbf {\bibinfo {volume} {120}},\ \bibinfo {pages} {e2302572120} (\bibinfo {year} {2023}{\natexlab{a}})}\BibitemShut {NoStop}%
\bibitem [{\citenamefont {Wen}(2017)}]{wen_colloquium_2017}%
  \BibitemOpen
  \bibfield  {author} {\bibinfo {author} {\bibfnamefont {X.-G.}\ \bibnamefont {Wen}},\ }\href {https://doi.org/10.1103/RevModPhys.89.041004} {\bibfield  {journal} {\bibinfo  {journal} {Reviews of Modern Physics}\ }\textbf {\bibinfo {volume} {89}},\ \bibinfo {pages} {041004} (\bibinfo {year} {2017})}\BibitemShut {NoStop}%
\bibitem [{\citenamefont {Sieberer}\ \emph {et~al.}(2025)\citenamefont {Sieberer}, \citenamefont {Buchhold}, \citenamefont {Marino},\ and\ \citenamefont {Diehl}}]{sieberer_universality_2025}%
  \BibitemOpen
  \bibfield  {author} {\bibinfo {author} {\bibfnamefont {L.~M.}\ \bibnamefont {Sieberer}}, \bibinfo {author} {\bibfnamefont {M.}~\bibnamefont {Buchhold}}, \bibinfo {author} {\bibfnamefont {J.}~\bibnamefont {Marino}},\ and\ \bibinfo {author} {\bibfnamefont {S.}~\bibnamefont {Diehl}},\ }\href {https://doi.org/10.1103/RevModPhys.97.025004} {\bibfield  {journal} {\bibinfo  {journal} {Reviews of Modern Physics}\ }\textbf {\bibinfo {volume} {97}},\ \bibinfo {pages} {025004} (\bibinfo {year} {2025})}\BibitemShut {NoStop}%
\bibitem [{\citenamefont {Weidemann}\ \emph {et~al.}(2020)\citenamefont {Weidemann}, \citenamefont {Kremer}, \citenamefont {Helbig}, \citenamefont {Hofmann}, \citenamefont {Stegmaier}, \citenamefont {Greiter}, \citenamefont {Thomale},\ and\ \citenamefont {Szameit}}]{weidemann_topological_2020}%
  \BibitemOpen
  \bibfield  {author} {\bibinfo {author} {\bibfnamefont {S.}~\bibnamefont {Weidemann}}, \bibinfo {author} {\bibfnamefont {M.}~\bibnamefont {Kremer}}, \bibinfo {author} {\bibfnamefont {T.}~\bibnamefont {Helbig}}, \bibinfo {author} {\bibfnamefont {T.}~\bibnamefont {Hofmann}}, \bibinfo {author} {\bibfnamefont {A.}~\bibnamefont {Stegmaier}}, \bibinfo {author} {\bibfnamefont {M.}~\bibnamefont {Greiter}}, \bibinfo {author} {\bibfnamefont {R.}~\bibnamefont {Thomale}},\ and\ \bibinfo {author} {\bibfnamefont {A.}~\bibnamefont {Szameit}},\ }\href {https://doi.org/10.1126/science.aaz8727} {\bibfield  {journal} {\bibinfo  {journal} {Science}\ }\textbf {\bibinfo {volume} {368}},\ \bibinfo {pages} {311} (\bibinfo {year} {2020})}\BibitemShut {NoStop}%
\bibitem [{\citenamefont {Schreiber}\ \emph {et~al.}(2015)\citenamefont {Schreiber}, \citenamefont {Hodgman}, \citenamefont {Bordia}, \citenamefont {L{\"u}schen}, \citenamefont {Fischer}, \citenamefont {Vosk}, \citenamefont {Altman}, \citenamefont {Schneider},\ and\ \citenamefont {Bloch}}]{schreiber_observation_2015}%
  \BibitemOpen
  \bibfield  {author} {\bibinfo {author} {\bibfnamefont {M.}~\bibnamefont {Schreiber}}, \bibinfo {author} {\bibfnamefont {S.~S.}\ \bibnamefont {Hodgman}}, \bibinfo {author} {\bibfnamefont {P.}~\bibnamefont {Bordia}}, \bibinfo {author} {\bibfnamefont {H.~P.}\ \bibnamefont {L{\"u}schen}}, \bibinfo {author} {\bibfnamefont {M.~H.}\ \bibnamefont {Fischer}}, \bibinfo {author} {\bibfnamefont {R.}~\bibnamefont {Vosk}}, \bibinfo {author} {\bibfnamefont {E.}~\bibnamefont {Altman}}, \bibinfo {author} {\bibfnamefont {U.}~\bibnamefont {Schneider}},\ and\ \bibinfo {author} {\bibfnamefont {I.}~\bibnamefont {Bloch}},\ }\href {https://doi.org/10.1126/science.aaa7432} {\bibfield  {journal} {\bibinfo  {journal} {Science}\ }\textbf {\bibinfo {volume} {349}},\ \bibinfo {pages} {842} (\bibinfo {year} {2015})}\BibitemShut {NoStop}%
\bibitem [{\citenamefont {Daley}(2014)}]{daley_quantum_2014}%
  \BibitemOpen
  \bibfield  {author} {\bibinfo {author} {\bibfnamefont {A.~J.}\ \bibnamefont {Daley}},\ }\href {https://doi.org/10.1080/00018732.2014.933502} {\bibfield  {journal} {\bibinfo  {journal} {Advances in Physics}\ }\textbf {\bibinfo {volume} {63}},\ \bibinfo {pages} {77} (\bibinfo {year} {2014})}\BibitemShut {NoStop}%
\bibitem [{\citenamefont {Bender}\ and\ \citenamefont {Boettcher}(1998)}]{bender_real_1998}%
  \BibitemOpen
  \bibfield  {author} {\bibinfo {author} {\bibfnamefont {C.~M.}\ \bibnamefont {Bender}}\ and\ \bibinfo {author} {\bibfnamefont {S.}~\bibnamefont {Boettcher}},\ }\href {https://doi.org/10.1103/PhysRevLett.80.5243} {\bibfield  {journal} {\bibinfo  {journal} {Physical Review Letters}\ }\textbf {\bibinfo {volume} {80}},\ \bibinfo {pages} {5243} (\bibinfo {year} {1998})}\BibitemShut {NoStop}%
\bibitem [{\citenamefont {Ashida}\ \emph {et~al.}(2020)\citenamefont {Ashida}, \citenamefont {Gong},\ and\ \citenamefont {Ueda}}]{ashida_non-hermitian_2020}%
  \BibitemOpen
  \bibfield  {author} {\bibinfo {author} {\bibfnamefont {Y.}~\bibnamefont {Ashida}}, \bibinfo {author} {\bibfnamefont {Z.}~\bibnamefont {Gong}},\ and\ \bibinfo {author} {\bibfnamefont {M.}~\bibnamefont {Ueda}},\ }\href {https://doi.org/10.1080/00018732.2021.1876991} {\bibfield  {journal} {\bibinfo  {journal} {Advances in Physics}\ }\textbf {\bibinfo {volume} {69}},\ \bibinfo {pages} {249} (\bibinfo {year} {2020})}\BibitemShut {NoStop}%
\bibitem [{\citenamefont {Bergholtz}\ \emph {et~al.}(2021)\citenamefont {Bergholtz}, \citenamefont {Budich},\ and\ \citenamefont {Kunst}}]{bergholtz_exceptional_2021}%
  \BibitemOpen
  \bibfield  {author} {\bibinfo {author} {\bibfnamefont {E.~J.}\ \bibnamefont {Bergholtz}}, \bibinfo {author} {\bibfnamefont {J.~C.}\ \bibnamefont {Budich}},\ and\ \bibinfo {author} {\bibfnamefont {F.~K.}\ \bibnamefont {Kunst}},\ }\href {https://doi.org/10.1103/RevModPhys.93.015005} {\bibfield  {journal} {\bibinfo  {journal} {Reviews of Modern Physics}\ }\textbf {\bibinfo {volume} {93}},\ \bibinfo {pages} {015005} (\bibinfo {year} {2021})}\BibitemShut {NoStop}%
\bibitem [{\citenamefont {Banerjee}\ \emph {et~al.}(2023{\natexlab{b}})\citenamefont {Banerjee}, \citenamefont {Sarkar}, \citenamefont {Dey},\ and\ \citenamefont {Narayan}}]{banerjee_non-hermitian_2023}%
  \BibitemOpen
  \bibfield  {author} {\bibinfo {author} {\bibfnamefont {A.}~\bibnamefont {Banerjee}}, \bibinfo {author} {\bibfnamefont {R.}~\bibnamefont {Sarkar}}, \bibinfo {author} {\bibfnamefont {S.}~\bibnamefont {Dey}},\ and\ \bibinfo {author} {\bibfnamefont {A.}~\bibnamefont {Narayan}},\ }\href {https://doi.org/10.1088/1361-648X/acd1cb} {\bibfield  {journal} {\bibinfo  {journal} {Journal of Physics: Condensed Matter}\ }\textbf {\bibinfo {volume} {35}},\ \bibinfo {pages} {333001} (\bibinfo {year} {2023}{\natexlab{b}})}\BibitemShut {NoStop}%
\bibitem [{\citenamefont {Lee}\ \emph {et~al.}(2019)\citenamefont {Lee}, \citenamefont {Ahn}, \citenamefont {Zhou},\ and\ \citenamefont {Vishwanath}}]{lee_topological_2019}%
  \BibitemOpen
  \bibfield  {author} {\bibinfo {author} {\bibfnamefont {J.~Y.}\ \bibnamefont {Lee}}, \bibinfo {author} {\bibfnamefont {J.}~\bibnamefont {Ahn}}, \bibinfo {author} {\bibfnamefont {H.}~\bibnamefont {Zhou}},\ and\ \bibinfo {author} {\bibfnamefont {A.}~\bibnamefont {Vishwanath}},\ }\href {https://doi.org/10.1103/PhysRevLett.123.206404} {\bibfield  {journal} {\bibinfo  {journal} {Physical Review Letters}\ }\textbf {\bibinfo {volume} {123}},\ \bibinfo {pages} {206404} (\bibinfo {year} {2019})}\BibitemShut {NoStop}%
\bibitem [{\citenamefont {Kato}(1966)}]{kato_perturbation_1966}%
  \BibitemOpen
  \bibfield  {author} {\bibinfo {author} {\bibfnamefont {T.}~\bibnamefont {Kato}},\ }\href@noop {} {\emph {\bibinfo {title} {Perturbation {{Theory}} for {{Linear Operators}}}}}\ (\bibinfo  {publisher} {Springer},\ \bibinfo {address} {New York},\ \bibinfo {year} {1966})\BibitemShut {NoStop}%
\bibitem [{\citenamefont {Yao}\ and\ \citenamefont {Wang}(2018)}]{yao_edge_2018}%
  \BibitemOpen
  \bibfield  {author} {\bibinfo {author} {\bibfnamefont {S.}~\bibnamefont {Yao}}\ and\ \bibinfo {author} {\bibfnamefont {Z.}~\bibnamefont {Wang}},\ }\href {https://doi.org/10.1103/PhysRevLett.121.086803} {\bibfield  {journal} {\bibinfo  {journal} {Physical Review Letters}\ }\textbf {\bibinfo {volume} {121}},\ \bibinfo {pages} {086803} (\bibinfo {year} {2018})}\BibitemShut {NoStop}%
\bibitem [{\citenamefont {Gohsrich}\ \emph {et~al.}(2025)\citenamefont {Gohsrich}, \citenamefont {Banerjee},\ and\ \citenamefont {Kunst}}]{gohsrich_non-hermitian_2025-1}%
  \BibitemOpen
  \bibfield  {author} {\bibinfo {author} {\bibfnamefont {J.~T.}\ \bibnamefont {Gohsrich}}, \bibinfo {author} {\bibfnamefont {A.}~\bibnamefont {Banerjee}},\ and\ \bibinfo {author} {\bibfnamefont {F.~K.}\ \bibnamefont {Kunst}},\ }\href {https://doi.org/10.1209/0295-5075/addf77} {\bibfield  {journal} {\bibinfo  {journal} {Europhysics Letters}\ }\textbf {\bibinfo {volume} {150}},\ \bibinfo {pages} {60001} (\bibinfo {year} {2025})}\BibitemShut {NoStop}%
\bibitem [{\citenamefont {Kunst}\ \emph {et~al.}(2018)\citenamefont {Kunst}, \citenamefont {Edvardsson}, \citenamefont {Budich},\ and\ \citenamefont {Bergholtz}}]{kunst_biorthogonal_2018}%
  \BibitemOpen
  \bibfield  {author} {\bibinfo {author} {\bibfnamefont {F.~K.}\ \bibnamefont {Kunst}}, \bibinfo {author} {\bibfnamefont {E.}~\bibnamefont {Edvardsson}}, \bibinfo {author} {\bibfnamefont {J.~C.}\ \bibnamefont {Budich}},\ and\ \bibinfo {author} {\bibfnamefont {E.~J.}\ \bibnamefont {Bergholtz}},\ }\href {https://doi.org/10.1103/PhysRevLett.121.026808} {\bibfield  {journal} {\bibinfo  {journal} {Physical Review Letters}\ }\textbf {\bibinfo {volume} {121}},\ \bibinfo {pages} {026808} (\bibinfo {year} {2018})}\BibitemShut {NoStop}%
\bibitem [{\citenamefont {Zhang}\ \emph {et~al.}(2020)\citenamefont {Zhang}, \citenamefont {Yang},\ and\ \citenamefont {Fang}}]{zhang_correspondence_2020}%
  \BibitemOpen
  \bibfield  {author} {\bibinfo {author} {\bibfnamefont {K.}~\bibnamefont {Zhang}}, \bibinfo {author} {\bibfnamefont {Z.}~\bibnamefont {Yang}},\ and\ \bibinfo {author} {\bibfnamefont {C.}~\bibnamefont {Fang}},\ }\href {https://doi.org/10.1103/PhysRevLett.125.126402} {\bibfield  {journal} {\bibinfo  {journal} {Physical Review Letters}\ }\textbf {\bibinfo {volume} {125}},\ \bibinfo {pages} {126402} (\bibinfo {year} {2020})}\BibitemShut {NoStop}%
\bibitem [{\citenamefont {Panda}\ and\ \citenamefont {Banerjee}(2020)}]{panda_entanglement_2020}%
  \BibitemOpen
  \bibfield  {author} {\bibinfo {author} {\bibfnamefont {A.}~\bibnamefont {Panda}}\ and\ \bibinfo {author} {\bibfnamefont {S.}~\bibnamefont {Banerjee}},\ }\href {https://doi.org/10.1103/PhysRevB.101.184201} {\bibfield  {journal} {\bibinfo  {journal} {Physical Review B}\ }\textbf {\bibinfo {volume} {101}},\ \bibinfo {pages} {184201} (\bibinfo {year} {2020})}\BibitemShut {NoStop}%
\bibitem [{\citenamefont {Banerjee}\ \emph {et~al.}(2022)\citenamefont {Banerjee}, \citenamefont {Hegde}, \citenamefont {Agarwala},\ and\ \citenamefont {Narayan}}]{banerjee_chiral_2022}%
  \BibitemOpen
  \bibfield  {author} {\bibinfo {author} {\bibfnamefont {A.}~\bibnamefont {Banerjee}}, \bibinfo {author} {\bibfnamefont {S.~S.}\ \bibnamefont {Hegde}}, \bibinfo {author} {\bibfnamefont {A.}~\bibnamefont {Agarwala}},\ and\ \bibinfo {author} {\bibfnamefont {A.}~\bibnamefont {Narayan}},\ }\href {https://doi.org/10.1103/PhysRevB.105.205403} {\bibfield  {journal} {\bibinfo  {journal} {Physical Review B}\ }\textbf {\bibinfo {volume} {105}},\ \bibinfo {pages} {205403} (\bibinfo {year} {2022})}\BibitemShut {NoStop}%
\bibitem [{\citenamefont {Chang}\ \emph {et~al.}(2020)\citenamefont {Chang}, \citenamefont {You}, \citenamefont {Wen},\ and\ \citenamefont {Ryu}}]{chang_entanglement_2020}%
  \BibitemOpen
  \bibfield  {author} {\bibinfo {author} {\bibfnamefont {P.-Y.}\ \bibnamefont {Chang}}, \bibinfo {author} {\bibfnamefont {J.-S.}\ \bibnamefont {You}}, \bibinfo {author} {\bibfnamefont {X.}~\bibnamefont {Wen}},\ and\ \bibinfo {author} {\bibfnamefont {S.}~\bibnamefont {Ryu}},\ }\href {https://doi.org/10.1103/PhysRevResearch.2.033069} {\bibfield  {journal} {\bibinfo  {journal} {Physical Review Research}\ }\textbf {\bibinfo {volume} {2}},\ \bibinfo {pages} {033069} (\bibinfo {year} {2020})}\BibitemShut {NoStop}%
\bibitem [{\citenamefont {Kitaev}(2009)}]{kitaev_periodic_2009}%
  \BibitemOpen
  \bibfield  {author} {\bibinfo {author} {\bibfnamefont {A.}~\bibnamefont {Kitaev}},\ }\href {https://doi.org/10.1063/1.3149495} {\bibfield  {journal} {\bibinfo  {journal} {AIP Conference Proceedings}\ }\textbf {\bibinfo {volume} {1134}},\ \bibinfo {pages} {22} (\bibinfo {year} {2009})}\BibitemShut {NoStop}%
\bibitem [{\citenamefont {Ryu}\ \emph {et~al.}(2010)\citenamefont {Ryu}, \citenamefont {Schnyder}, \citenamefont {Furusaki},\ and\ \citenamefont {Ludwig}}]{ryu_topological_2010}%
  \BibitemOpen
  \bibfield  {author} {\bibinfo {author} {\bibfnamefont {S.}~\bibnamefont {Ryu}}, \bibinfo {author} {\bibfnamefont {A.~P.}\ \bibnamefont {Schnyder}}, \bibinfo {author} {\bibfnamefont {A.}~\bibnamefont {Furusaki}},\ and\ \bibinfo {author} {\bibfnamefont {A.~W.~W.}\ \bibnamefont {Ludwig}},\ }\href {https://doi.org/10.1088/1367-2630/12/6/065010} {\bibfield  {journal} {\bibinfo  {journal} {New Journal of Physics}\ }\textbf {\bibinfo {volume} {12}},\ \bibinfo {pages} {065010} (\bibinfo {year} {2010})}\BibitemShut {NoStop}%
\bibitem [{\citenamefont {Shen}\ \emph {et~al.}(2018)\citenamefont {Shen}, \citenamefont {Zhen},\ and\ \citenamefont {Fu}}]{shen_topological_2018}%
  \BibitemOpen
  \bibfield  {author} {\bibinfo {author} {\bibfnamefont {H.}~\bibnamefont {Shen}}, \bibinfo {author} {\bibfnamefont {B.}~\bibnamefont {Zhen}},\ and\ \bibinfo {author} {\bibfnamefont {L.}~\bibnamefont {Fu}},\ }\href {https://doi.org/10.1103/PhysRevLett.120.146402} {\bibfield  {journal} {\bibinfo  {journal} {Physical Review Letters}\ }\textbf {\bibinfo {volume} {120}},\ \bibinfo {pages} {146402} (\bibinfo {year} {2018})}\BibitemShut {NoStop}%
\bibitem [{\citenamefont {Kawabata}\ \emph {et~al.}(2019)\citenamefont {Kawabata}, \citenamefont {Shiozaki}, \citenamefont {Ueda},\ and\ \citenamefont {Sato}}]{kawabata_symmetry_2019}%
  \BibitemOpen
  \bibfield  {author} {\bibinfo {author} {\bibfnamefont {K.}~\bibnamefont {Kawabata}}, \bibinfo {author} {\bibfnamefont {K.}~\bibnamefont {Shiozaki}}, \bibinfo {author} {\bibfnamefont {M.}~\bibnamefont {Ueda}},\ and\ \bibinfo {author} {\bibfnamefont {M.}~\bibnamefont {Sato}},\ }\href {https://doi.org/10.1103/PhysRevX.9.041015} {\bibfield  {journal} {\bibinfo  {journal} {Physical Review X}\ }\textbf {\bibinfo {volume} {9}},\ \bibinfo {pages} {041015} (\bibinfo {year} {2019})}\BibitemShut {NoStop}%
\bibitem [{\citenamefont {Wang}\ \emph {et~al.}(2021)\citenamefont {Wang}, \citenamefont {Dutt}, \citenamefont {Wojcik},\ and\ \citenamefont {Fan}}]{wang_topological_2021-1}%
  \BibitemOpen
  \bibfield  {author} {\bibinfo {author} {\bibfnamefont {K.}~\bibnamefont {Wang}}, \bibinfo {author} {\bibfnamefont {A.}~\bibnamefont {Dutt}}, \bibinfo {author} {\bibfnamefont {C.~C.}\ \bibnamefont {Wojcik}},\ and\ \bibinfo {author} {\bibfnamefont {S.}~\bibnamefont {Fan}},\ }\href {https://doi.org/10.1038/s41586-021-03848-x} {\bibfield  {journal} {\bibinfo  {journal} {Nature}\ }\textbf {\bibinfo {volume} {598}},\ \bibinfo {pages} {59} (\bibinfo {year} {2021})}\BibitemShut {NoStop}%
\bibitem [{\citenamefont {K{\"o}nig}\ \emph {et~al.}(2023)\citenamefont {K{\"o}nig}, \citenamefont {Yang}, \citenamefont {Budich},\ and\ \citenamefont {Bergholtz}}]{konig_braid-protected_2023}%
  \BibitemOpen
  \bibfield  {author} {\bibinfo {author} {\bibfnamefont {J.~L.~K.}\ \bibnamefont {K{\"o}nig}}, \bibinfo {author} {\bibfnamefont {K.}~\bibnamefont {Yang}}, \bibinfo {author} {\bibfnamefont {J.~C.}\ \bibnamefont {Budich}},\ and\ \bibinfo {author} {\bibfnamefont {E.~J.}\ \bibnamefont {Bergholtz}},\ }\href {https://doi.org/10.1103/PhysRevResearch.5.L042010} {\bibfield  {journal} {\bibinfo  {journal} {Physical Review Research}\ }\textbf {\bibinfo {volume} {5}},\ \bibinfo {pages} {L042010} (\bibinfo {year} {2023})}\BibitemShut {NoStop}%
\bibitem [{\citenamefont {Yoshida}\ \emph {et~al.}(2025)\citenamefont {Yoshida}, \citenamefont {K{\"o}nig}, \citenamefont {R{\o}dland}, \citenamefont {Bergholtz},\ and\ \citenamefont {St{\aa}lhammar}}]{yoshida_winding_2025}%
  \BibitemOpen
  \bibfield  {author} {\bibinfo {author} {\bibfnamefont {T.}~\bibnamefont {Yoshida}}, \bibinfo {author} {\bibfnamefont {J.~L.~K.}\ \bibnamefont {K{\"o}nig}}, \bibinfo {author} {\bibfnamefont {L.}~\bibnamefont {R{\o}dland}}, \bibinfo {author} {\bibfnamefont {E.~J.}\ \bibnamefont {Bergholtz}},\ and\ \bibinfo {author} {\bibfnamefont {M.}~\bibnamefont {St{\aa}lhammar}},\ }\href {https://doi.org/10.1103/PhysRevResearch.7.L012021} {\bibfield  {journal} {\bibinfo  {journal} {Physical Review Research}\ }\textbf {\bibinfo {volume} {7}},\ \bibinfo {pages} {L012021} (\bibinfo {year} {2025})}\BibitemShut {NoStop}%
\bibitem [{\citenamefont {St{\aa}lhammar}\ and\ \citenamefont {R{\o}dland}(2025)}]{stalhammar_abelian_2025}%
  \BibitemOpen
  \bibfield  {author} {\bibinfo {author} {\bibfnamefont {M.}~\bibnamefont {St{\aa}lhammar}}\ and\ \bibinfo {author} {\bibfnamefont {L.}~\bibnamefont {R{\o}dland}},\ }\href {https://doi.org/10.1103/8lyx-vtnw} {\bibfield  {journal} {\bibinfo  {journal} {Physical Review Research}\ }\textbf {\bibinfo {volume} {7}},\ \bibinfo {pages} {033246} (\bibinfo {year} {2025})}\BibitemShut {NoStop}%
\bibitem [{\citenamefont {Montag}\ \emph {et~al.}(2026)\citenamefont {Montag}, \citenamefont {Felski},\ and\ \citenamefont {Kunst}}]{montag_spectral_2026}%
  \BibitemOpen
  \bibfield  {author} {\bibinfo {author} {\bibfnamefont {A.}~\bibnamefont {Montag}}, \bibinfo {author} {\bibfnamefont {A.}~\bibnamefont {Felski}},\ and\ \bibinfo {author} {\bibfnamefont {F.~K.}\ \bibnamefont {Kunst}},\ }\href {https://doi.org/10.21468/SciPostPhys.20.5.133} {\bibfield  {journal} {\bibinfo  {journal} {SciPost Physics}\ }\textbf {\bibinfo {volume} {20}},\ \bibinfo {pages} {133} (\bibinfo {year} {2026})}\BibitemShut {NoStop}%
\bibitem [{\citenamefont {Lein}(2020)}]{lein_choosing_2020}%
  \BibitemOpen
  \bibfield  {author} {\bibinfo {author} {\bibfnamefont {M.}~\bibnamefont {Lein}},\ }\href@noop {} {} (\bibinfo {year} {2020}),\ \Eprint {https://arxiv.org/abs/2010.09261} {arXiv:2010.09261} \BibitemShut {NoStop}%
\bibitem [{\citenamefont {Morse}(1934)}]{morse_calculus_1934}%
  \BibitemOpen
  \bibfield  {author} {\bibinfo {author} {\bibfnamefont {M.}~\bibnamefont {Morse}},\ }\href {https://doi.org/10.1090/coll/018} {\emph {\bibinfo {title} {The {{Calculus}} of {{Variations}} in the {{Large}}}}},\ \bibinfo {series} {Colloquium {{Publications}}}, Vol.~\bibinfo {volume} {18}\ (\bibinfo  {publisher} {American Mathematical Society},\ \bibinfo {year} {1934})\BibitemShut {NoStop}%
\bibitem [{\citenamefont {Milnor}(1963)}]{milnor_morse_1963}%
  \BibitemOpen
  \bibfield  {author} {\bibinfo {author} {\bibfnamefont {J.~W.}\ \bibnamefont {Milnor}},\ }\href@noop {} {\emph {\bibinfo {title} {Morse Theory}}},\ \bibinfo {series} {Annals of {{Mathematics Studies}}}, Vol.~\bibinfo {volume} {51}\ (\bibinfo  {publisher} {Princeton University Press},\ \bibinfo {year} {1963})\BibitemShut {NoStop}%
\bibitem [{\citenamefont {Montag}\ and\ \citenamefont {Kunst}(2024)}]{montag_essential_2024-1}%
  \BibitemOpen
  \bibfield  {author} {\bibinfo {author} {\bibfnamefont {A.}~\bibnamefont {Montag}}\ and\ \bibinfo {author} {\bibfnamefont {F.~K.}\ \bibnamefont {Kunst}},\ }\href {https://doi.org/10.1063/5.0206211} {\bibfield  {journal} {\bibinfo  {journal} {Journal of Mathematical Physics}\ }\textbf {\bibinfo {volume} {65}},\ \bibinfo {pages} {122101} (\bibinfo {year} {2024})}\BibitemShut {NoStop}%
\bibitem [{\citenamefont {Nielsen}\ and\ \citenamefont {Ninomiya}(1981)}]{nielsen_no-go_1981}%
  \BibitemOpen
  \bibfield  {author} {\bibinfo {author} {\bibfnamefont {H.~B.}\ \bibnamefont {Nielsen}}\ and\ \bibinfo {author} {\bibfnamefont {M.}~\bibnamefont {Ninomiya}},\ }\href {https://doi.org/10.1016/0370-2693(81)91026-1} {\bibfield  {journal} {\bibinfo  {journal} {Physics Letters B}\ }\textbf {\bibinfo {volume} {105}},\ \bibinfo {pages} {219} (\bibinfo {year} {1981})}\BibitemShut {NoStop}%
\bibitem [{\citenamefont {Ying}\ and\ \citenamefont {Kamenev}(2018)}]{ying_symmetry-protected_2018}%
  \BibitemOpen
  \bibfield  {author} {\bibinfo {author} {\bibfnamefont {X.}~\bibnamefont {Ying}}\ and\ \bibinfo {author} {\bibfnamefont {A.}~\bibnamefont {Kamenev}},\ }\href {https://doi.org/10.1103/PhysRevLett.121.086810} {\bibfield  {journal} {\bibinfo  {journal} {Physical Review Letters}\ }\textbf {\bibinfo {volume} {121}},\ \bibinfo {pages} {086810} (\bibinfo {year} {2018})}\BibitemShut {NoStop}%
\bibitem [{\citenamefont {Bessho}\ and\ \citenamefont {Sato}(2021)}]{bessho_nielsen-ninomiya_2021}%
  \BibitemOpen
  \bibfield  {author} {\bibinfo {author} {\bibfnamefont {T.}~\bibnamefont {Bessho}}\ and\ \bibinfo {author} {\bibfnamefont {M.}~\bibnamefont {Sato}},\ }\href {https://doi.org/10.1103/PhysRevLett.127.196404} {\bibfield  {journal} {\bibinfo  {journal} {Physical Review Letters}\ }\textbf {\bibinfo {volume} {127}},\ \bibinfo {pages} {196404} (\bibinfo {year} {2021})}\BibitemShut {NoStop}%
\bibitem [{\citenamefont {Peschel}(2003)}]{peschel_calculation_2003}%
  \BibitemOpen
  \bibfield  {author} {\bibinfo {author} {\bibfnamefont {I.}~\bibnamefont {Peschel}},\ }\href {https://doi.org/10.1088/0305-4470/36/14/101} {\bibfield  {journal} {\bibinfo  {journal} {Journal of Physics A: Mathematical and General}\ }\textbf {\bibinfo {volume} {36}},\ \bibinfo {pages} {L205} (\bibinfo {year} {2003})}\BibitemShut {NoStop}%
\bibitem [{\citenamefont {Calabrese}\ and\ \citenamefont {Cardy}(2009)}]{calabrese_entanglement_2009}%
  \BibitemOpen
  \bibfield  {author} {\bibinfo {author} {\bibfnamefont {P.}~\bibnamefont {Calabrese}}\ and\ \bibinfo {author} {\bibfnamefont {J.}~\bibnamefont {Cardy}},\ }\href {https://doi.org/10.1088/1751-8113/42/50/504005} {\bibfield  {journal} {\bibinfo  {journal} {Journal of Physics A: Mathematical and Theoretical}\ }\textbf {\bibinfo {volume} {42}},\ \bibinfo {pages} {504005} (\bibinfo {year} {2009})}\BibitemShut {NoStop}%
\bibitem [{\citenamefont {Calabrese}\ and\ \citenamefont {Cardy}(2004)}]{calabrese_entanglement_2004}%
  \BibitemOpen
  \bibfield  {author} {\bibinfo {author} {\bibfnamefont {P.}~\bibnamefont {Calabrese}}\ and\ \bibinfo {author} {\bibfnamefont {J.}~\bibnamefont {Cardy}},\ }\href {https://doi.org/10.1088/1742-5468/2004/06/P06002} {\bibfield  {journal} {\bibinfo  {journal} {Journal of Statistical Mechanics: Theory and Experiment}\ }\textbf {\bibinfo {volume} {2004}},\ \bibinfo {pages} {P06002} (\bibinfo {year} {2004})}\BibitemShut {NoStop}%
\bibitem [{\citenamefont {Hatano}\ and\ \citenamefont {Nelson}(1996)}]{hatano_localization_1996}%
  \BibitemOpen
  \bibfield  {author} {\bibinfo {author} {\bibfnamefont {N.}~\bibnamefont {Hatano}}\ and\ \bibinfo {author} {\bibfnamefont {D.~R.}\ \bibnamefont {Nelson}},\ }\href {https://doi.org/10.1103/PhysRevLett.77.570} {\bibfield  {journal} {\bibinfo  {journal} {Physical Review Letters}\ }\textbf {\bibinfo {volume} {77}},\ \bibinfo {pages} {570} (\bibinfo {year} {1996})}\BibitemShut {NoStop}%
\bibitem [{\citenamefont {Gohsrich}\ \emph {et~al.}(2024)\citenamefont {Gohsrich}, \citenamefont {Fauman},\ and\ \citenamefont {Kunst}}]{gohsrich_exceptional_2024}%
  \BibitemOpen
  \bibfield  {author} {\bibinfo {author} {\bibfnamefont {J.~T.}\ \bibnamefont {Gohsrich}}, \bibinfo {author} {\bibfnamefont {J.}~\bibnamefont {Fauman}},\ and\ \bibinfo {author} {\bibfnamefont {F.~K.}\ \bibnamefont {Kunst}},\ }\href@noop {} {} (\bibinfo {year} {2024}),\ \Eprint {https://arxiv.org/abs/2403.12018} {arXiv:2403.12018} \BibitemShut {NoStop}%
\bibitem [{\citenamefont {Marques}\ and\ \citenamefont {Dias}(2022)}]{marques_generalized_2022}%
  \BibitemOpen
  \bibfield  {author} {\bibinfo {author} {\bibfnamefont {A.~M.}\ \bibnamefont {Marques}}\ and\ \bibinfo {author} {\bibfnamefont {R.~G.}\ \bibnamefont {Dias}},\ }\href {https://doi.org/10.1103/PhysRevB.106.205146} {\bibfield  {journal} {\bibinfo  {journal} {Physical Review B}\ }\textbf {\bibinfo {volume} {106}},\ \bibinfo {pages} {205146} (\bibinfo {year} {2022})}\BibitemShut {NoStop}%
\bibitem [{\citenamefont {Zhang}\ \emph {et~al.}(2022)\citenamefont {Zhang}, \citenamefont {Denner}, \citenamefont {Bzdu{\v s}ek}, \citenamefont {Sentef},\ and\ \citenamefont {Neupert}}]{zhang_symmetry_2022-1}%
  \BibitemOpen
  \bibfield  {author} {\bibinfo {author} {\bibfnamefont {S.-B.}\ \bibnamefont {Zhang}}, \bibinfo {author} {\bibfnamefont {M.~M.}\ \bibnamefont {Denner}}, \bibinfo {author} {\bibfnamefont {T.}~\bibnamefont {Bzdu{\v s}ek}}, \bibinfo {author} {\bibfnamefont {M.~A.}\ \bibnamefont {Sentef}},\ and\ \bibinfo {author} {\bibfnamefont {T.}~\bibnamefont {Neupert}},\ }\href {https://doi.org/10.1103/PhysRevB.106.L121102} {\bibfield  {journal} {\bibinfo  {journal} {Physical Review B}\ }\textbf {\bibinfo {volume} {106}},\ \bibinfo {pages} {L121102} (\bibinfo {year} {2022})}\BibitemShut {NoStop}%
\bibitem [{\citenamefont {Kawabata}\ \emph {et~al.}(2022)\citenamefont {Kawabata}, \citenamefont {Shiozaki},\ and\ \citenamefont {Ryu}}]{kawabata_many-body_2022-1}%
  \BibitemOpen
  \bibfield  {author} {\bibinfo {author} {\bibfnamefont {K.}~\bibnamefont {Kawabata}}, \bibinfo {author} {\bibfnamefont {K.}~\bibnamefont {Shiozaki}},\ and\ \bibinfo {author} {\bibfnamefont {S.}~\bibnamefont {Ryu}},\ }\href {https://doi.org/10.1103/PhysRevB.105.165137} {\bibfield  {journal} {\bibinfo  {journal} {Physical Review B}\ }\textbf {\bibinfo {volume} {105}},\ \bibinfo {pages} {165137} (\bibinfo {year} {2022})}\BibitemShut {NoStop}%
\bibitem [{\citenamefont {Longhi}(2022)}]{longhi_self-healing_2022}%
  \BibitemOpen
  \bibfield  {author} {\bibinfo {author} {\bibfnamefont {S.}~\bibnamefont {Longhi}},\ }\href {https://doi.org/10.1103/PhysRevLett.128.157601} {\bibfield  {journal} {\bibinfo  {journal} {Physical Review Letters}\ }\textbf {\bibinfo {volume} {128}},\ \bibinfo {pages} {157601} (\bibinfo {year} {2022})}\BibitemShut {NoStop}%
\bibitem [{\citenamefont {Fitzpatrick}\ \emph {et~al.}(2017)\citenamefont {Fitzpatrick}, \citenamefont {Sundaresan}, \citenamefont {Li}, \citenamefont {Koch},\ and\ \citenamefont {Houck}}]{fitzpatrick_observation_2017}%
  \BibitemOpen
  \bibfield  {author} {\bibinfo {author} {\bibfnamefont {M.}~\bibnamefont {Fitzpatrick}}, \bibinfo {author} {\bibfnamefont {N.~M.}\ \bibnamefont {Sundaresan}}, \bibinfo {author} {\bibfnamefont {A.~C.~Y.}\ \bibnamefont {Li}}, \bibinfo {author} {\bibfnamefont {J.}~\bibnamefont {Koch}},\ and\ \bibinfo {author} {\bibfnamefont {A.~A.}\ \bibnamefont {Houck}},\ }\href {https://doi.org/10.1103/PhysRevX.7.011016} {\bibfield  {journal} {\bibinfo  {journal} {Physical Review X}\ }\textbf {\bibinfo {volume} {7}},\ \bibinfo {pages} {011016} (\bibinfo {year} {2017})}\BibitemShut {NoStop}%
\end{thebibliography}%


\begin{thebibliography}{15}%
\makeatletter
\providecommand \@ifxundefined [1]{%
 \@ifx{#1\undefined}
}%
\providecommand \@ifnum [1]{%
 \ifnum #1\expandafter \@firstoftwo
 \else \expandafter \@secondoftwo
 \fi
}%
\providecommand \@ifx [1]{%
 \ifx #1\expandafter \@firstoftwo
 \else \expandafter \@secondoftwo
 \fi
}%
\providecommand \natexlab [1]{#1}%
\providecommand \enquote  [1]{``#1''}%
\providecommand \bibnamefont  [1]{#1}%
\providecommand \bibfnamefont [1]{#1}%
\providecommand \citenamefont [1]{#1}%
\providecommand \href@noop [0]{\@secondoftwo}%
\providecommand \href [0]{\begingroup \@sanitize@url \@href}%
\providecommand \@href[1]{\@@startlink{#1}\@@href}%
\providecommand \@@href[1]{\endgroup#1\@@endlink}%
\providecommand \@sanitize@url [0]{\catcode `\\12\catcode `\$12\catcode `\&12\catcode `\#12\catcode `\^12\catcode `\_12\catcode `\%12\relax}%
\providecommand \@@startlink[1]{}%
\providecommand \@@endlink[0]{}%
\providecommand \url  [0]{\begingroup\@sanitize@url \@url }%
\providecommand \@url [1]{\endgroup\@href {#1}{\urlprefix }}%
\providecommand \urlprefix  [0]{URL }%
\providecommand \Eprint [0]{\href }%
\providecommand \doibase [0]{https://doi.org/}%
\providecommand \selectlanguage [0]{\@gobble}%
\providecommand \bibinfo  [0]{\@secondoftwo}%
\providecommand \bibfield  [0]{\@secondoftwo}%
\providecommand \translation [1]{[#1]}%
\providecommand \BibitemOpen [0]{}%
\providecommand \bibitemStop [0]{}%
\providecommand \bibitemNoStop [0]{.\EOS\space}%
\providecommand \EOS [0]{\spacefactor3000\relax}%
\providecommand \BibitemShut  [1]{\csname bibitem#1\endcsname}%
\let\auto@bib@innerbib\@empty
\bibitem [{\citenamefont {Hatano}\ and\ \citenamefont {Nelson}(1996)}]{hatano_localization_1996-si}%
  \BibitemOpen
  \bibfield  {author} {\bibinfo {author} {\bibfnamefont {N.}~\bibnamefont {Hatano}}\ and\ \bibinfo {author} {\bibfnamefont {D.~R.}\ \bibnamefont {Nelson}},\ }\href {https://doi.org/10.1103/PhysRevLett.77.570} {\bibfield  {journal} {\bibinfo  {journal} {Physical Review Letters}\ }\textbf {\bibinfo {volume} {77}},\ \bibinfo {pages} {570} (\bibinfo {year} {1996})}\BibitemShut {NoStop}%
\bibitem [{\citenamefont {Kawabata}\ \emph {et~al.}(2019)\citenamefont {Kawabata}, \citenamefont {Shiozaki}, \citenamefont {Ueda},\ and\ \citenamefont {Sato}}]{kawabata_symmetry_2019-si}%
  \BibitemOpen
  \bibfield  {author} {\bibinfo {author} {\bibfnamefont {K.}~\bibnamefont {Kawabata}}, \bibinfo {author} {\bibfnamefont {K.}~\bibnamefont {Shiozaki}}, \bibinfo {author} {\bibfnamefont {M.}~\bibnamefont {Ueda}},\ and\ \bibinfo {author} {\bibfnamefont {M.}~\bibnamefont {Sato}},\ }\href {https://doi.org/10.1103/PhysRevX.9.041015} {\bibfield  {journal} {\bibinfo  {journal} {Physical Review X}\ }\textbf {\bibinfo {volume} {9}},\ \bibinfo {pages} {041015} (\bibinfo {year} {2019})}\BibitemShut {NoStop}%
\bibitem [{\citenamefont {Ashida}\ \emph {et~al.}(2020)\citenamefont {Ashida}, \citenamefont {Gong},\ and\ \citenamefont {Ueda}}]{ashida_non-hermitian_2020-si}%
  \BibitemOpen
  \bibfield  {author} {\bibinfo {author} {\bibfnamefont {Y.}~\bibnamefont {Ashida}}, \bibinfo {author} {\bibfnamefont {Z.}~\bibnamefont {Gong}},\ and\ \bibinfo {author} {\bibfnamefont {M.}~\bibnamefont {Ueda}},\ }\href {https://doi.org/10.1080/00018732.2021.1876991} {\bibfield  {journal} {\bibinfo  {journal} {Advances in Physics}\ }\textbf {\bibinfo {volume} {69}},\ \bibinfo {pages} {249} (\bibinfo {year} {2020})}\BibitemShut {NoStop}%
\bibitem [{\citenamefont {Bergholtz}\ \emph {et~al.}(2021)\citenamefont {Bergholtz}, \citenamefont {Budich},\ and\ \citenamefont {Kunst}}]{bergholtz_exceptional_2021-si}%
  \BibitemOpen
  \bibfield  {author} {\bibinfo {author} {\bibfnamefont {E.~J.}\ \bibnamefont {Bergholtz}}, \bibinfo {author} {\bibfnamefont {J.~C.}\ \bibnamefont {Budich}},\ and\ \bibinfo {author} {\bibfnamefont {F.~K.}\ \bibnamefont {Kunst}},\ }\href {https://doi.org/10.1103/RevModPhys.93.015005} {\bibfield  {journal} {\bibinfo  {journal} {Reviews of Modern Physics}\ }\textbf {\bibinfo {volume} {93}},\ \bibinfo {pages} {015005} (\bibinfo {year} {2021})}\BibitemShut {NoStop}%
\bibitem [{\citenamefont {Gohsrich}\ \emph {et~al.}(2025)\citenamefont {Gohsrich}, \citenamefont {Banerjee},\ and\ \citenamefont {Kunst}}]{gohsrich_non-hermitian_2025-1-si}%
  \BibitemOpen
  \bibfield  {author} {\bibinfo {author} {\bibfnamefont {J.~T.}\ \bibnamefont {Gohsrich}}, \bibinfo {author} {\bibfnamefont {A.}~\bibnamefont {Banerjee}},\ and\ \bibinfo {author} {\bibfnamefont {F.~K.}\ \bibnamefont {Kunst}},\ }\href {https://doi.org/10.1209/0295-5075/addf77} {\bibfield  {journal} {\bibinfo  {journal} {Europhysics Letters}\ }\textbf {\bibinfo {volume} {150}},\ \bibinfo {pages} {60001} (\bibinfo {year} {2025})}\BibitemShut {NoStop}%
\bibitem [{\citenamefont {Lee}(2016)}]{lee_anomalous_2016-si}%
  \BibitemOpen
  \bibfield  {author} {\bibinfo {author} {\bibfnamefont {T.~E.}\ \bibnamefont {Lee}},\ }\href {https://doi.org/10.1103/PhysRevLett.116.133903} {\bibfield  {journal} {\bibinfo  {journal} {Physical Review Letters}\ }\textbf {\bibinfo {volume} {116}},\ \bibinfo {pages} {133903} (\bibinfo {year} {2016})}\BibitemShut {NoStop}%
\bibitem [{\citenamefont {Lieu}(2018)}]{lieu_topological_2018-si}%
  \BibitemOpen
  \bibfield  {author} {\bibinfo {author} {\bibfnamefont {S.}~\bibnamefont {Lieu}},\ }\href {https://doi.org/10.1103/PhysRevB.97.045106} {\bibfield  {journal} {\bibinfo  {journal} {Physical Review B}\ }\textbf {\bibinfo {volume} {97}},\ \bibinfo {pages} {045106} (\bibinfo {year} {2018})}\BibitemShut {NoStop}%
\bibitem [{\citenamefont {Kunst}\ \emph {et~al.}(2018)\citenamefont {Kunst}, \citenamefont {Edvardsson}, \citenamefont {Budich},\ and\ \citenamefont {Bergholtz}}]{kunst_biorthogonal_2018-si}%
  \BibitemOpen
  \bibfield  {author} {\bibinfo {author} {\bibfnamefont {F.~K.}\ \bibnamefont {Kunst}}, \bibinfo {author} {\bibfnamefont {E.}~\bibnamefont {Edvardsson}}, \bibinfo {author} {\bibfnamefont {J.~C.}\ \bibnamefont {Budich}},\ and\ \bibinfo {author} {\bibfnamefont {E.~J.}\ \bibnamefont {Bergholtz}},\ }\href {https://doi.org/10.1103/PhysRevLett.121.026808} {\bibfield  {journal} {\bibinfo  {journal} {Physical Review Letters}\ }\textbf {\bibinfo {volume} {121}},\ \bibinfo {pages} {026808} (\bibinfo {year} {2018})}\BibitemShut {NoStop}%
\bibitem [{\citenamefont {Yao}\ and\ \citenamefont {Wang}(2018)}]{yao_edge_2018-si}%
  \BibitemOpen
  \bibfield  {author} {\bibinfo {author} {\bibfnamefont {S.}~\bibnamefont {Yao}}\ and\ \bibinfo {author} {\bibfnamefont {Z.}~\bibnamefont {Wang}},\ }\href {https://doi.org/10.1103/PhysRevLett.121.086803} {\bibfield  {journal} {\bibinfo  {journal} {Physical Review Letters}\ }\textbf {\bibinfo {volume} {121}},\ \bibinfo {pages} {086803} (\bibinfo {year} {2018})}\BibitemShut {NoStop}%
\bibitem [{\citenamefont {Yin}\ \emph {et~al.}(2018)\citenamefont {Yin}, \citenamefont {Jiang}, \citenamefont {Li}, \citenamefont {L{\"u}},\ and\ \citenamefont {Chen}}]{yin_geometrical_2018-si}%
  \BibitemOpen
  \bibfield  {author} {\bibinfo {author} {\bibfnamefont {C.}~\bibnamefont {Yin}}, \bibinfo {author} {\bibfnamefont {H.}~\bibnamefont {Jiang}}, \bibinfo {author} {\bibfnamefont {L.}~\bibnamefont {Li}}, \bibinfo {author} {\bibfnamefont {R.}~\bibnamefont {L{\"u}}},\ and\ \bibinfo {author} {\bibfnamefont {S.}~\bibnamefont {Chen}},\ }\href {https://doi.org/10.1103/PhysRevA.97.052115} {\bibfield  {journal} {\bibinfo  {journal} {Physical Review A}\ }\textbf {\bibinfo {volume} {97}},\ \bibinfo {pages} {052115} (\bibinfo {year} {2018})}\BibitemShut {NoStop}%
\bibitem [{\citenamefont {Kato}(1966)}]{kato_perturbation_1966-si}%
  \BibitemOpen
  \bibfield  {author} {\bibinfo {author} {\bibfnamefont {T.}~\bibnamefont {Kato}},\ }\href@noop {} {\emph {\bibinfo {title} {Perturbation {{Theory}} for {{Linear Operators}}}}}\ (\bibinfo  {publisher} {Springer},\ \bibinfo {address} {New York},\ \bibinfo {year} {1966})\BibitemShut {NoStop}%
\bibitem [{\citenamefont {Gohsrich}\ \emph {et~al.}(2024)\citenamefont {Gohsrich}, \citenamefont {Fauman},\ and\ \citenamefont {Kunst}}]{gohsrich_exceptional_2024-si}%
  \BibitemOpen
  \bibfield  {author} {\bibinfo {author} {\bibfnamefont {J.~T.}\ \bibnamefont {Gohsrich}}, \bibinfo {author} {\bibfnamefont {J.}~\bibnamefont {Fauman}},\ and\ \bibinfo {author} {\bibfnamefont {F.~K.}\ \bibnamefont {Kunst}},\ }\href@noop {} {} (\bibinfo {year} {2024}),\ \Eprint {https://arxiv.org/abs/2403.12018} {arXiv:2403.12018} \BibitemShut {NoStop}%
\bibitem [{\citenamefont {Marques}\ and\ \citenamefont {Dias}(2022)}]{marques_generalized_2022-si}%
  \BibitemOpen
  \bibfield  {author} {\bibinfo {author} {\bibfnamefont {A.~M.}\ \bibnamefont {Marques}}\ and\ \bibinfo {author} {\bibfnamefont {R.~G.}\ \bibnamefont {Dias}},\ }\href {https://doi.org/10.1103/PhysRevB.106.205146} {\bibfield  {journal} {\bibinfo  {journal} {Physical Review B}\ }\textbf {\bibinfo {volume} {106}},\ \bibinfo {pages} {205146} (\bibinfo {year} {2022})}\BibitemShut {NoStop}%
\bibitem [{\citenamefont {McCann}(2026)}]{mccann_non-hermitian_2026-si}%
  \BibitemOpen
  \bibfield  {author} {\bibinfo {author} {\bibfnamefont {E.}~\bibnamefont {McCann}},\ }\href {https://doi.org/10.1103/hdzd-94qm} {\bibfield  {journal} {\bibinfo  {journal} {Physical Review B}\ }\textbf {\bibinfo {volume} {113}},\ \bibinfo {pages} {045403} (\bibinfo {year} {2026})}\BibitemShut {NoStop}%
\bibitem [{\citenamefont {Montag}\ and\ \citenamefont {Kunst}(2024)}]{montag_essential_2024-1-si}%
  \BibitemOpen
  \bibfield  {author} {\bibinfo {author} {\bibfnamefont {A.}~\bibnamefont {Montag}}\ and\ \bibinfo {author} {\bibfnamefont {F.~K.}\ \bibnamefont {Kunst}},\ }\href {https://doi.org/10.1063/5.0206211} {\bibfield  {journal} {\bibinfo  {journal} {Journal of Mathematical Physics}\ }\textbf {\bibinfo {volume} {65}},\ \bibinfo {pages} {122101} (\bibinfo {year} {2024})}\BibitemShut {NoStop}%
\end{thebibliography}%

\clearpage
\onecolumngrid

\renewcommand{\theequation}{S.\arabic{equation}}
\renewcommand{\thefigure}{S.\arabic{figure}}
\renewcommand{\thesection}{S.\Roman{section}}
\setcounter{equation}{0}
\setcounter{figure}{0}
\setcounter{section}{0}


\begin{center}
\textbf{\large Supplementary Information for:\\[2pt]
Spectral-topology-induced criticality in non-Hermitian fermionic metals}
\end{center}

\vspace{0.5cm}

To supplement the findings of the main text, we provide further models illustrating different facets of the described phenomenology.
First, we discuss the Hatano--Nelson model~\citesupp{hatano_localization_1996-si} to illustrate that the strict condition of having time-reversal symmetry~\citesupp{kawabata_symmetry_2019-si} is not necessary to define the dynamical topological index and the notion of half-filling, as the weaker spectral constraint $\{E\}=\{E^*\}$ is sufficient.
Furthermore, we illustrate that other spectral constraints may be employed.
Second, we discuss the anisotropic Su--Schrieffer--Heeger (SSH) model~\citesupp{ashida_non-hermitian_2020-si,bergholtz_exceptional_2021-si,gohsrich_non-hermitian_2025-1-si,lee_anomalous_2016-si,lieu_topological_2018-si,kunst_biorthogonal_2018-si,yao_edge_2018-si,yin_geometrical_2018-si} to illustrate the multi-band case, where the occurrence of exceptional points (EPs)~\citesupp{kato_perturbation_1966-si,ashida_non-hermitian_2020-si,bergholtz_exceptional_2021-si} enriches the phenomenology compared to the single-band case.
As the anisotropic SSH model shows a non-generic feature, namely, two closed spectral curves merging into a single closed spectral curve and vice versa, we also discuss a generalized SSH model~\citesupp{gohsrich_exceptional_2024-si,marques_generalized_2022-si,mccann_non-hermitian_2026-si}, where this is not the case.
We also use that model to illustrate the extended Brillouin zone scheme, which is necessary to define the dynamical topological index in the multi-band case.
And finally, we discuss two distinct coupled Hatano--Nelson models, one having time-reversal symmetry featuring topological phase transitions without EPs in a multi-band model, and the other one being pseudo-Hermitian featuring regions (instead of points) in parameter space, where the dynamical topological index is ill-defined.

For later use, let us define pseudo-Hermiticity, which entails $\mathcal{PT}$-symmetry and pseudo-Hermitian symmetry~\citesupp{montag_essential_2024-1-si}.
It is given by $\mathcal{H}(k)=\eta \mathcal{H}^\dagger(k) \eta^{-1}$ with $\eta$ being Hermitian and invertible.

\section{The Hatano--Nelson model}

The first model we discuss is the Hatano--Nelson model~\citesupp{hatano_localization_1996}.
Under periodic boundary conditions, it is given by
\begin{equation}
    \label{eq:hn}
    H^\mathrm{HN} = \big( \sum_{n=1}^{N-1} t_{1} c_{n}^\dagger c_{n+1} + t_{-1} c_{n+1}^\dagger c_{n} \big) + t_{1} c_{N}^\dagger c_{1} + t_{-1} c_{1}^\dagger c_{N},
\end{equation}
and has the Bloch Hamiltonian 
\begin{equation}
    \label{eq:hn-bloch}
    \mathcal{H}^\mathrm{HN}(k)=t_{1}\ee^{\ii k} + t_{-1}\ee^{-\ii k}.
\end{equation}
Since it is a single-band model, the energy band $E(k)$ directly corresponds to Bloch Hamiltonian, i.e., $E(k)=\mathcal{H}(k)$.
As discussed in the main text, for real $t_1$ and $t_{-1}$, this model is time-reversal symmetric.
That is $\mathcal{H}(-k)=\mathcal{T} \mathcal{H}^*(k) \mathcal{T}^\dagger$ with a unitary $\mathcal{T}$ satisfying $\mathcal{T}\mathcal{T}^*=\pm\mathbb{1}$ and the identity matrix~$\mathbb{1}$, corresponding to TRS in Ref.~\citesupp{kawabata_symmetry_2019-si}, not TRS$^\dagger$.
For the single-band case, $\mathcal{T}=1$, i.e., $\mathcal{H}(-k)=\mathcal{H}^*(k)$.

Furthermore, let us define pseudo-Hermiticity, which entails $\mathcal{PT}$-symmetry and pseudo-Hermitian symmetry~\citesupp{montag_essential_2024-1-si}.
It is given by $\mathcal{H}(k)=\eta \mathcal{H}^\dagger(k) \eta^{-1}$ with $\eta$ being Hermitian and invertible.
For single-band models, this reduces to $\mathcal{H}(k)=\mathcal{H}^\dagger(k)=\mathcal{H}^*(k)$, i.e., $\mathcal{H}(k)$ is real, which we do not cover as the dynamical topological index is ill-defined in this case.

In the main text, we discuss that the spectral constraint $\{E\} = \{E^*\}$ is enough to have a well-defined dynamical topological index and a reasonable notion of half-filling.
To exemplify this, we substitute $t_1 \to t_1 \ee^{\ii \phi}$ and $t_{-1} \to t_{-1} \ee^{-\ii\phi}$ in Eq.~\eqref{eq:hn-bloch} yielding $H(k,\phi)=t_{1}\ee^{\ii (k+\phi)} + t_{-1}\ee^{-\ii (k+\phi)}=E(k,\phi)$, and we restrict to real $t_1$, $t_{-1}$ and $\phi$. 
Changing $\phi$ leaves the full spectrum invariant, thus satisfying the spectral constraint $\{E\} = \{E^*\}$, while the system is neither time-reversal symmetric nor pseudo-Hermitian.
Figure~\ref{fig:hn}(a,b) illustrates this.

We want to also mention that other spectral constraints may be employed to have a reasonable notion of half-filling.
For example, $\{E\} = \{-E^*\}$ also results in a spectrum where the same number of states are above and below $\Im E(k)=0$.
This is exemplified in Fig.~\ref{fig:hn}(c,d), where we have introduced the complex hoppings $t_1 \to t_1 \ee^{\ii \varphi}$ and $t_{-1} \to t_{-1} \ee^{\ii\varphi}$.
This results in a spectrum, which is rotated by $\varphi$ around the origin in the complex plane, compared to the corresponding real hoppings, i.e., for $\varphi=0$.

\clearpage
\section{The SSH model and a generalized SSH model}

To discuss properties of multi-band models, let us introduce the anisotropic SSH model~\citesupp{ashida_non-hermitian_2020-si,bergholtz_exceptional_2021-si,gohsrich_non-hermitian_2025-1-si,lee_anomalous_2016-si,lieu_topological_2018-si,kunst_biorthogonal_2018-si,yao_edge_2018-si,yin_geometrical_2018-si}, and a generalized version of it, akin to the generalized Hatano--Nelson model employed in the main text~\citesupp{gohsrich_exceptional_2024-si}, featuring a generalized chiral symmetry~\citesupp{marques_generalized_2022-si}.

It is given by 
\begin{equation}
\label{eq:gen-ssh}
    H = \sum_n \left[ (t_1 + \gamma)\, a_{n}^\dagger b_{n+l-1} + t_2\, a_{n}^\dagger b_{n-r} + (t_1 - \gamma)\, b_{n}^\dagger a_{n-r+1} + t_2\, b_{n}^\dagger a_{n+l} \right],
\end{equation}
where $a_n^\dagger$ ($b_n^\dagger$) creates a fermion on sublattice $A$ ($B$) of the $n$-th unit cell, $t_1\pm \gamma$ denote anisotropic hopping amplitudes with $\gamma$ controlling the anisotropy, and $t_2$ is an isotropic hopping amplitude.
The integers $r$ and $l$ specify the range of rightward and leftward hoppings, respectively, such that $r=l=1$ recovers the conventional anisotropic SSH model.
In this case, $t_1\pm \gamma$ correspond to intracell hoppings and $t_2$ to intercell hoppings.
We note that this model is distinct from the SSH model with generalized chiral symmetry discussed in Ref.~\citenum{mccann_non-hermitian_2026-si}.

The Bloch Hamiltonian corresponding to Eq.~\eqref{eq:gen-ssh} reads
\begin{equation}
\mathcal{H}(k) =
\begin{pmatrix}
0 & (t_1 + \gamma)\, \ee^{\ii (l-1)k} + t_2\, \ee^{-\ii r k} \\
(t_1 - \gamma)\, \ee^{-\ii (r-1)k} + t_2\, \ee^{\ii l k} & 0
\end{pmatrix},
\end{equation}
which is time-reversal symmetric, with $\mathcal{T}=\mathbb{1}_2$ being the $2\times2$ identity matrix, for real $t_1$, $t_2$ and $\gamma$.

\subsection{The SSH model}

Figure~\ref{fig:ssh} shows the spectrum of the anisotropic SSH model ($r=l=1$) for $t_1=1$, $t_2=1/2$ and different $\gamma>0$.
Fig.~\ref{fig:ssh}(a) shows the complex spectrum with $\gamma=1/4$, which features two distinct closed spectral curves, each contributing with one maximum to the dynamical topological index, such that $\nudyn=2$ over all.
Panel (b) shows the imaginary part of the dispersion.

By increasing to $\gamma=1/2$, depicted in Fig.~\ref{fig:ssh}(c,d), one finds an EP at $k=\pi$ where the two spectral curves touch.
Due to the non-Hermitian degeneracy at the EP, the bands are non-Morse functions and thus $\nudyn$ is ill-defined, hinting at the potential for a topological phase transition.

Passing through this EP, the complex spectrum features a single loop, depicted in Fig.~\ref{fig:ssh}(e) with $\gamma=1$.
In this case, we find $\nudyn=1$, see also Fig.~\ref{fig:ssh}(f), showing the contributing maximum at $k=\pi$.
Thus, we tuned through a topological phase transition, enabled by a non-Morse degeneracy due to an EP.

Increasing $\gamma$ further yields another EP at $\gamma=3/2$, depicted in Fig.~\ref{fig:ssh}(g,h).
Again, this EP signals a topological phase transition, splitting the single closed spectral curve in Fig.~\ref{fig:ssh}(e) into two, as depicted in Fig.~\ref{fig:ssh}(i) with $\gamma = 7/4$.
This time, one of the two closed spectral curves is entirely above $\Im E(k)=0$, resulting in $\nudyn=0$.
Hence, due to this imaginary line gap~\citesupp{kawabata_symmetry_2019-si}, the system does not feature Fermi points and scale-invariant gapless excitations.
Thus, for this parameter choice, the system is not critical according to our definition in the main text.

\subsection{The generalized SSH model}

Let us turn our attention to the generalized SSH model with $l=2$ and $r=1$, illustrated in Fig.~\ref{fig:gen-ssh}.
As anticipated, we discuss this model to illustrate that while an EP signals a topological phase transition, it does not necessarily need to split or merge closed spectral curves as for the anisotropic SSH model.
Furthermore, we use this model to discuss the picture of going to an extended Brillouin zone scheme.

In Fig.~\ref{fig:gen-ssh}(a), we show the complex spectrum featuring the $4$-fold spectral symmetry as a result of the generalized chiral symmetry.
It is a single loop, and contributes with three maxima and two minima to a dynamical topological index of $\nudyn=1$.
This can also be inferred from panel (b), where we show the corresponding imaginary part of the dispersion over the first Brillouin zone.
In panel (c), we go to an extended Brillouin zone scheme by additionally showing the second Brillouin zone.
By picking one of the two functions, e.g., the green solid curve, one essentially combines the two imaginary bands in (b) into a smooth and periodic parametrization of the imaginary part in (a).
As this parametrization yields a Morse function, this reduces the problem to the single-band case.

Increasing to $\gamma=1$ yields EPs at $k=\pm\pi/2$, as shown in Fig.~\ref{fig:gen-ssh}(d,e).
As the individual complex bands are not smooth, the extended Brillouin zone scheme does not yield a smooth parametrization, which in turn prohibits the determination of the dynamical topological index and indicates a spectral phase transition.

In Fig.~\ref{fig:gen-ssh}(f-h), we increased to $\gamma=3/2$, where three maxima contribute to $\nudyn=3$, and thus show that the EP has enabled a topological phase transition.
Comparing (b), (e) and (g), we see the mechanism discussed in the main text and illustrated in Fig. 2(b) unfolding:
By increasing $\gamma$, the two minima around $k\approx\pm\pi/2$ contributing to $\nudyn$ in (b) `annihilate' the maxima below $\Im E(k)=0$ not contributing to $\nudyn$ via the non-Hermitian degeneracy at the EP in (e), increasing the dynamical topological index in (g) by two.

\section{Coupled Hatano--Nelson models}

Lastly, we want to discuss two different coupled Hatano--Nelson models with different symmetry-constraints.

\subsection{Model 1}

The first model has the Bloch-Hamiltonian
\begin{equation}
    \label{eq:c-hn-1}
    \mathcal{H}^{(1)}(k) = \begin{pmatrix}
        \mathcal{H}^\mathrm{HN}(k)+ \ii \gamma & \kappa \\
        \kappa & \mathcal{H}^\mathrm{HN}(k)- \ii \gamma
    \end{pmatrix},
\end{equation}
where $\gamma$ encodes onsite gain and loss, and $\kappa$ couples the two chains.
This model is time-reversal symmetric with $\mathcal{T}=\sigma_x$, where~$\sigma_x$ is the Pauli-$x$ matrix, for real parameters.
Spectra and imaginary parts of the dispersions are depicted in Fig.~\ref{fig:c-hn}(a-f), and we vary the onsite gain and loss $\gamma=5/4,\sqrt{2},2$ while fixing $\kappa=1$.

In Fig.~\ref{fig:c-hn}(a), with $\gamma=5/4$, the two closed spectral loops are intersecting, however, the imaginary part of the dispersion in panel (b) shows that the system is non-degenerate.
Each imaginary band contributes with a maximum at $k=\pi/2$ to $\nudyn=2$.

Increasing to $\gamma=\sqrt{2}$ makes the two spectral curves touch in Fig.~\ref{fig:c-hn}(c) at the origin.
Still, this does not correspond to a degeneracy as seen in panel (d), in contrast to the anisotropic SSH model discussed before, where the touching corresponds to a degeneracy.
In panel (d), the blue band touches $\Im E(k)=0$ at $k=-\pi/2$ and the orange band at $k=+\pi/2$.
As these points correspond to non-simple zeros, the dynamical topological index is not defined, hinting at a topological phase transition without an EP.
This is confirmed by further increasing to $\gamma=2$, illustrated in Fig.~\ref{fig:c-hn}(e,f).
The system is not critical for this parameter choice as $\nudyn=0$, similar to the anisotropic SSH model in Fig.~\ref{fig:ssh}(i,j).

\subsection{Model 2}
The second model we discuss has the Bloch Hamiltonian
\begin{equation}
    \label{eq:c-hn-2}
    \mathcal{H}^{(2)}(k) = \begin{pmatrix}
        \mathcal{H}^\mathrm{HN}(k)+ \ii \gamma & \kappa \\
        \kappa & \mathcal{H}^\mathrm{HN}(-k)- \ii \gamma
    \end{pmatrix},
\end{equation}
which closely resembles model 1, with the distinction that the bottom right matrix element is $\mathcal{H}^\mathrm{HN}(-k)$ instead of $\mathcal{H}^\mathrm{HN}(k)$.
Due to this change, the system is pseudo-Hermitian with $\eta=\sigma_x$ instead of time-reversal symmetric for real parameters.
Spectra and imaginary parts of the dispersions are depicted in Fig.~\ref{fig:c-hn}(g-l), where we vary the coupling $\kappa=1/2,1,3/2$ while fixing the onsite gain and loss $\gamma=1$.

For $\kappa=1/2$, the spectrum in Fig.~\ref{fig:c-hn}(g) resembles the one for the time-reversal symmetric model, Fig.~\ref{fig:c-hn}(e).
However, due to the different symmetry constraint, the imaginary part of the dispersion is qualitatively different, as the two imaginary bands in panel~(h) satisfy $f_1(k) = - f_2(k)$ compared to $f_1(k) = - f_2(-k)$ in Fig.~\ref{fig:c-hn}(f).
Again, the system is not critical, as $\nudyn=0$.

Increasing the coupling yields an EP at $\kappa=1$ and $k=\pi/2$, and the corresponding spectrum is depicted in Fig.~\ref{fig:c-hn}(i).
Interestingly, diagonalizing the corresponding Bloch Hamiltonian away from $k=\pi/2$ yields the imaginary bands in panel (j), which feature removable discontinuities at $k=\pi/2$.
Removing these removable discontinuities yields a smooth parametrization of the imaginary part of the dispersion.
This allows us to determine $\nudyn=1$, even though we technically do not define the dynamical topological index in the main text when $\mathcal{H}(k)$ is not diagonalizable for all $k$, indicating that some assumptions may be relaxed.

Further increasing the couplings, for example $\kappa=3/2$ as depicted in Fig.~\ref{fig:c-hn}(k-l), the system is not diagonalizable, and the bands depicted in Fig.~\ref{fig:c-hn}(l) are not smooth and thus non-Morse, rendering the dynamical topological index ill-defined.

\clearpage
\section*{Supplementary Figures}

\begin{figure*}[h]
    \centering
    \includegraphics[]{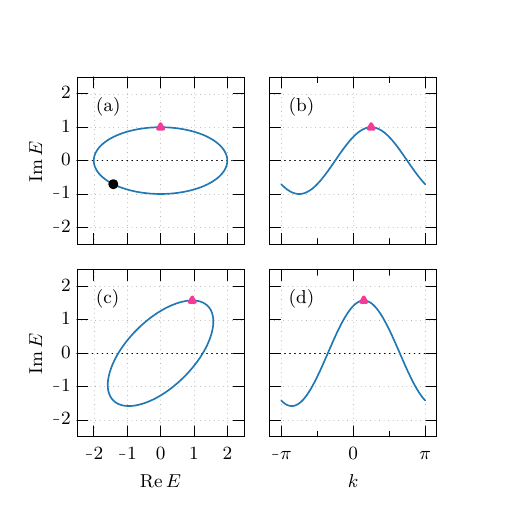}
    \caption{
    Spectra of the Hatano--Nelson model for different parameters.
    (a)~Complex spectrum and (b)~imaginary dispersion with $t_1=1.5\ee^{\ii\pi/4}$ and $t_{-1}=0.5\ee^{-\ii\pi/4}$.
    The black dot in (a) corresponds to $k=-\pi$.
    Even though the system does not feature time-reversal symmetry or pseudo-Hermiticity, it satisfies the spectral constraint $\{E\} = \{E^*\}$, leading to a well-defined dynamical topological index and notion of half-filling.
    (c)~Complex spectrum and (d)~imaginary dispersion with $t_1=1.5\ee^{\ii\pi/4}$ and $t_{-1}=0.5\ee^{\ii\pi/4}$, which corresponds to the spectrum with $t_1=1.5$ and $t_{-1}=0.5$ rotated by $\pi/4$ in the complex plane. 
    This is captured by the spectral constraint $\{E\} = \{-E^*\}$, which illustrates that other spectral constraints can lead to a well-defined dynamical topological index and a reasonable notion of half-filling.
    In all panels, the pink triangles correspond to the maxima contributing to the dynamical topological index, yielding $\nudyn=1$ in all depicted cases.
    }
    \label{fig:hn}
\end{figure*}

\begin{figure*}
    \centering
    \includegraphics[]{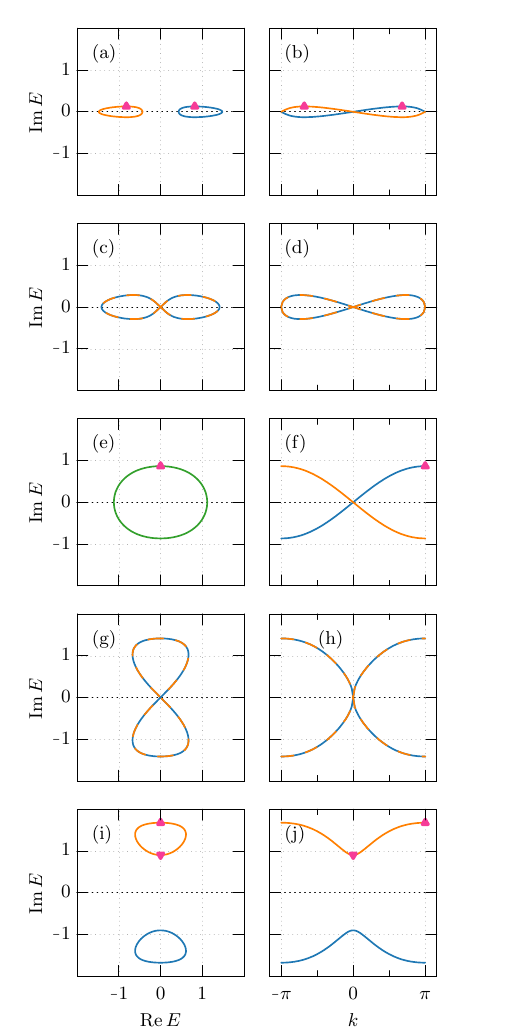}
    \caption{
    Spectra of the non-Hermitian SSH model for different parameters.
    The left column (a,c,e,g,i) shows complex spectra and the right column (b,d,f,h,j) shows imaginary dispersions, with $t_1=1$, $t_2=1/2$, and $\gamma=1/4,1/2,1,3/2,7/4$, respectively.
    (a,b) For $\gamma=1/4$, the spectrum (a) consists of two loops in orange and blue, which correspond to the two bands in the imaginary dispersion in (b).
    Again in all panels, pink triangles mark extrema contributing to the dynamical topological index, in this case, $\nudyn=2$.
    (c,d) At an EP at $\gamma=1/2$, the Hamiltonian is non-diagonalizable at $k=\pi$, which is indicated by blue-orange-dashed lines.
    This manifests in (d), as the bands are non-differentiable at this point, breaking smoothness and thus making the functions non-Morse.
    The dynamical topological index is ill-defined in this case.
    (e,f) Between the two EPs, e.g., at $\gamma=1$ as illustrated, the spectrum (e) consists of a single loop shown in green.
    The bands in (f) show one distinct maximum, thus $\nudyn=1$.
    (g,h) At the second EP at $\gamma=3/2$, the bands in (h) are non-differentiable at $k=0$.
    (i,j) For $\gamma=7/4$, the spectrum (i) again consists of two loops with smooth imaginary dispersions.
    As there is one maximum and one minimum above zero, also shown in (j), the dynamical topological index vanished, corresponding to a non-critical phase.
    }
    \label{fig:ssh}
\end{figure*}

\begin{figure*}
    \centering
    \includegraphics[]{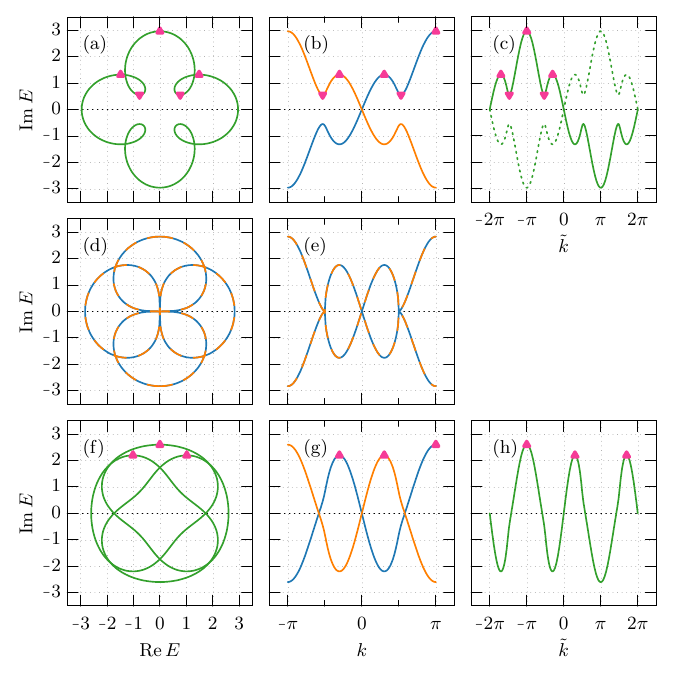}
    \caption{
    Spectra of a generalized SSH model for different parameters.
    The left column (a,d,f) shows complex spectra, all having a $4$-fold rotational symmetry as a result of a generalized chiral symmetry, and the center column (b,e,g) the imaginary dispersions with $l=2$, $r=1$, $t_1=1$, $t_2=2$, and $\gamma = 1/2,1,3/2$, respectively.
    The right column (c,h) corresponds to the extended Brillouin zone of (b,g), respectively.
    (a,b,c)~For $\gamma=1/2$, the complex spectrum in (a) consists of a single loop.
    Both bands in (b) are Morse functions.
    Employing an extended Brillouin zone scheme, in this case, showing the first two Brillouin zones in (c), one gets two bands, depicted as green solid or green dashed curves, which are Morse functions.
    By choosing any of those, the problem effectively reduces to the single-band case.
    Together with the real dispersion (not depicted), this yields a parametrization of the full spectrum in (a).
    In this case $\nudyn=1$.
    (d,e) At the EP at $\gamma=1$, while retaining the spectral symmetry in (d), the band structure is not smooth and thus does not allow determining the dynamical topological index.
    (f,g,h) For $\gamma=3/2$, the dynamical topological index changed to $\nudyn=3$, however, the spectrum (f) again consists of a single loop with smooth bands in (g) and smooth bands in the extended Brillouin zone (only one depicted) in (h).
    In all panels, the pink triangles correspond to the extrema contributing to the dynamical topological index.
    }
    \label{fig:gen-ssh}
\end{figure*}

\begin{figure*}
    \centering
    \includegraphics[]{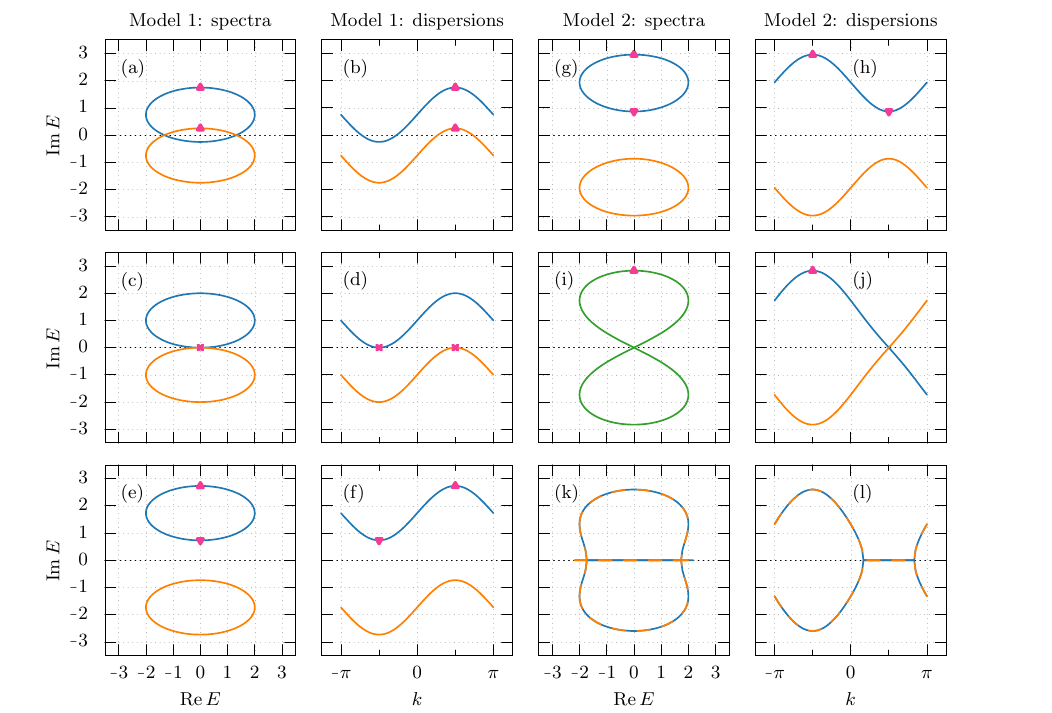}
    \caption{
    Complex spectra of the two different coupled Hatano--Nelson models.
    For both models we set $t_1=3/2$ and $t_{-1}=1/2$.
    The first column (a,c,e) shows complex spectra and the second column (b,d,f) the imaginary dispersions of model 1, Eq.~\eqref{eq:c-hn-1}, with coupling $\kappa=1$ and onsite gain and loss $\gamma=5/4,\sqrt{2},2$, respectively.
    Increasing $\gamma$ induces a topological phase transition at $\gamma=\sqrt{2}$ without an exceptional point.
    This phase-transition is marked by two non-simple zeros at $k=\pm\pi/2$ in (d), which are indicated by pink crosses in (c,d).
    The dynamical topological index is $\nudyn=2$ in (a,b), ill-defined in (c,d), and $\nudyn=0$ in (e,f).
    In the latter case, the system is not critical according to the definition in the main text.
    The third column (g,i,k) shows complex spectra and the last column (h,j,l) the imaginary dispersions of model 2, Eq.~\eqref{eq:c-hn-2}, with onsite gain and loss $\gamma=2$ and coupling $\kappa=1/2,1,3/2$, respectively.
    Similar to panels (e,f), the system is not critical as $\nudyn=0$.
    For $\kappa=1$ in (i,j), the system exhibits an EP at $k=\pi/2$, thus technically disallows determining the dynamical topological index as defined in the main text.
    However, this point corresponds to a removable discontinuity at $k=\pi/2$ in the imaginary part of the band structure in (j), and removing it yields smooth imaginary bands, which in turn allows determining $\nudyn=1$.
    Increasing $\kappa$ further, for example $\kappa=3/2$ as depicted in (k,l), the system stays defective.
    In this case, the bands in (l) are not smooth, and thus, the dynamical topological index is ill-defined.
    }
    \label{fig:c-hn}
\end{figure*}

\clearpage

\bibliographystylesupp{apsrev4-2}
\bibliographysupp{references-si.bib}

\end{document}